\documentclass[showpacs,aps]{revtex4}

\usepackage{amsmath,amssymb}
\usepackage{amsthm}        
\usepackage{curves}

\usepackage{amsmath,amssymb}
\usepackage{amsthm}

\begin{document}

\voffset=0.8cm
\setlength{\textheight}{23.5cm}

\theoremstyle{definition}\newtheorem{defin}{Definition}
\theoremstyle{definition}\newtheorem{theorem}{Theorem} 
\theoremstyle{definition}\newtheorem{proposition}{Proposition} 
\theoremstyle{definition}\newtheorem{lem}[defin]{Lemma} 

\title{On the hyperbolicity of Maxwell's equations with a local 
constitutive law}
\author{Volker Perlick}
\altaffiliation[Author's address since November 2010: ]{ZARM, 
Univ. Bremen, 28359 Bremen, Germany. Email: perlick@zarm.uni-bremen.de}
\affiliation{
Physics Department, Lancaster University,  Lancaster LA1 4YB, UK \\
and \\ The Cockcroft Institute, Warrington WA4 4AD, UK 
}



\begin{abstract}
\noindent
  Maxwell's equations are considered in metric-free form, with a local
  but otherwise arbitrary constitutive law. After splitting Maxwell's
  equations into evolution equations and constraints, we derive the 
  characteristic equation and we discuss its properties in detail. We 
  present several results that are relevant for the 
  question of whether the evolution equations are hyperbolic, strongly 
  hyperbolic or symmmetric hyperbolic. In particular, we
  give a convenient characterisation of all constitutive laws for which
  the evolution equations are symmetric hyperbolic. The latter 
  property is sufficient, but not necessary, for well-posedness of 
  the initial-value problem. By way of example, we illustrate our 
  results with the constitutive laws of biisotropic media and
  of Born-Infeld theory.   
\end{abstract}


\maketitle


\section{Introduction}\label{sec:intro}
\noindent
The fact that Maxwell's equations can be formulated on a bare manifold 
that need not carry a metric or a connection was first observed by Kottler \cite{Kottler1922}
and Cartan \cite{Cartan1986}, and later also by van Dantzig \cite{Dantzig1934} 
and Schroedinger \cite{Schroedinger1950}, p. 24. It plays a central part 
in Post's \cite{Post1962} systematic study of the formal structure of 
electromagnetism. More recently, Hehl and Obukhov \cite{HehlObukhov2003}
have studied the metric-free (or pre-metric) approach to Maxwell's equations and its 
physical implications in great detail. It is the philosophy of Hehl and
Obukhov to consider electromagnetism as more fundamental than gravity.
In their approach, Maxwell's equations are formulated on a bare manifold.
The constitutive law which connects the electromagnetic field strength
with the electromagnetic excitation plays the role of a ``space-time 
relation''. In other words, the space-time geometry, which governs 
gravity, is coded in the constitutive law. By formulating a particular,
very special, constitutive law one recovers a Lorentzian metric (at least
up to a conformal factor) and, thus, the ordinary general-relativistic 
theory of gravity. It is then very natural to speculate that more general 
constitutive laws, which lead to more general geometric structures, could
be considered as more general (hypothetical) theories of gravity. In
particular, such generalised theories of gravity typically predict 
birefringence for light rays in vacuo (i.e., under the influence of
gravity alone.) 

The Hehl-Obukhov approach has some similarities, although 
more in philosophy than in mathematical technicality, with an idea of 
Newman and his collaborators (see, e.g. \cite{FrittelliKozamehNewman1995})
who suggest to view the equation of wave fronts as fundamental for gravity. 
A characteristic feature (and, maybe, a drawback) of both approaches is 
that the Lorentzian metric that is supposed to describe gravity can be 
fixed only up to a conformal factor, i.e., up to a strictly positive but 
otherwise undetermined scalar function.  

The metric-free approach to Maxwell's equations gives a strong
motivation for investigating which constitutive laws are physically
reasonable and which ones are not. Here we want to discuss a criterion
which is not mentioned in the book by Hehl and Obukhov: We want to
characterise constitutive laws that yield a well-posed initial-value
problem. A closely related property is the admittance of wavelike
solutions, in a sense that is made precise in Section \ref{sec:planewave}
below. Our results apply to the case that the constitutive law is
interpreted as a spacetime relation (i.e., as the vacuum constitutive 
law in a generalised spacetime theory), but also to constitutive laws 
in a medium on a standard general-relativistic spacetime.  

In contrast to Hehl and Obukhov, who restrict to local and linear 
constitutive laws throughout, we allow for nonlinear constitutive 
laws. However, we have to maintain the restriction to local 
constitutive laws which excludes, e.g., media 
with memory such as ferromagnets. The latter case would lead to
integro-differential equations, for which an initial-value problem
in the standard sense cannot be formulated, whereas a local constitutive 
law leads to first-order differential equations for appropriately chosen
field components. In Section \ref{sec:maxwell} we derive these 
differential equations and we decompose them into evolution 
equations and constraints. In Sections \ref{sec:planewave} and
\ref{sec:characteristic} we derive and discuss the characteristic 
equation. The real roots of the characteristic equation determine the 
directions into which wavelike solutions can travel and, thereby, the 
``light cones'' of the theory. Non-real roots are associated with
``evanescent modes'', i.e., with exponentially decaying solutions.
In Section \ref{sec:hyperbolicity} we discuss the notions of hyperbolicity,
strong hyperbolicity and symmetric hyperbolicity. Hyperbolicity 
requires that all roots of the characteristic equation are real, i.e.,
that evanescent modes do not occur. Strong hyperbolicity is a 
necessary and sufficient condition for the initial-value problem
to be well-posed. Symmetric hyperbolicity is a 
sufficient but not a necessary condition for the initial-value problem
to be well-posed. Having established these notions, we investigate
some properties of the light cones for the case of hyperbolicity 
in Section \ref{sec:cone}. In Section
\ref{sec:trafo} we prove that the light cones are coordinate invariant
which is not obvious from our derivation. The following three sections
present some results that are useful for calculations: In Section 
\ref{sec:roots} we discuss how the roots of the characteristic equation
can actually be determined; in Section \ref{sec:SL2R} we demonstrate
the invariance of the characteristic equation under certain changes of
the constitutive law; and in Section \ref{sec:reduction} we derive an alternative
form of the characteristic equation. The case that the light cones are invariant
under temporal or spatial inversion is considered in Section \ref{sec:time},
and the case that there is no birefringence is considered in 
Section \ref{sec:non-birefringence}. In Section \ref{sec:symhyp} we
characterise the class of all constitutive laws for which the evolution 
equations are
symmetric hyperbolic. Finally, two examples are worked out in Section \ref{sec:ex}: biisotropic media and Born-Infeld electrodynamics. In
the conclusions we summarise the results that have been achieved so
far, and we list some important questions that are still open.



\section{Maxwell's equations in metric-free form}\label{sec:maxwell}
\noindent
We consider a 4-dimenional bare manifold, with coordinates $x = 
(x^0, x^1, x^2, x^3)$. We use Einstein's summation convention for 
latin indices running from 0 to 3 and for greek indices running from 
1 to 3. We refer to $x^0$ as to the time coordinate and to $x^{\mu}$
as to the spatial coordinates. At present, this is just a convenient 
mode of expression. As we have no structure on our manifold, it does
not make sense to ask whether the $x^0-$lines are timelike or whether
the hypersurfaces $x^0= \mathrm{constant}$ are spacelike. Later, 
however, we will discuss the question of whether initial values for
the evolution part of Maxwell's equations on the hypersurfaces 
$x^0= \mathrm{constant}$ determine a unique solution on an appropriate 
neighborhood. If this is true, one might view the covector $dx^0$ as 
``timelike'', in a sense determined by the evolution equations and not 
by a background structure. The covectors which are timelike in this 
sense turn out to form an open convex cone at each point of the manifold, 
see Section \ref{sec:cone} below. 

In standard index notation, Maxwell's equations read
\begin{equation}\label{eq:maxwell}
  \partial _{[a} F_{bc]} (x) = M_{abc} (x) \; , \qquad
  \partial _{[a} H_{bc]} (x) = J_{abc} (x) \; ,
\end{equation}
where the square bracket denotes antisymmetrization. Here
$F_{ab}=-F_{ba}$ is the electromagnetic field strength, 
$H_{ab}= -H_{ba}$ is the electromagnetic excitation, 
$J_{abc}$ is the electric current and $M_{abc}$ is a 
hypothetical magnetic current. (On physical grounds, there is good 
reason to assume that the latter is zero; however, we take it into
account for the sake of generality.) Under coordinate transformations, 
these fields change according to
\begin{equation}\label{eq:even}
\tilde{F} _{ab} \, = \, 
\frac{\partial x^d}{\partial \tilde{x}{}^a} \, 
\frac{\partial x^e}{\partial \tilde{x}{}^b} \,
F_{de} \, , \qquad
\tilde{M} _{abc} \, = \, 
\frac{\partial x^d}{\partial \tilde{x}{}^a} \, 
\frac{\partial x^e}{\partial \tilde{x}{}^b} \,
\frac{\partial x^f}{\partial \tilde{x}{}^c} \,
M_{def} \, , 
\end{equation}
\begin{equation}\label{eq:odd}
\tilde{H} _{ab} \, = \, 
\frac{\mathrm{det} \big( \frac{\partial x}{\partial \tilde{x}} \big)}{
\big| \mathrm{det} \big( \frac{\partial x}{\partial \tilde{x}} \big) \big|} \,
\frac{\partial x^d}{\partial \tilde{x}{}^a} \, 
\frac{\partial x^e}{\partial \tilde{x}{}^b} \,
H_{de} \, , \qquad
\tilde{J} _{abc} \, = \, 
\frac{\mathrm{det} \big( \frac{\partial x}{\partial \tilde{x}} \big)}{
\big| \mathrm{det} \big( \frac{\partial x}{\partial \tilde{x}} \big) \big|} \,
\frac{\partial x^d}{\partial \tilde{x}{}^a} \, 
\frac{\partial x^e}{\partial \tilde{x}{}^b} \,
\frac{\partial x^f}{\partial \tilde{x}{}^c} \,
J_{def} \, . 
\end{equation}
Thus, in the terminology of de Rham \cite{Rham1984}, $F_{ab}$ and 
$M_{abc}$ are even differential forms whereas $H_{ab}$ and  $J_{abc}$ 
are odd differential forms. (An even differential form is the 
same as a totally antisymmetric covariant tensor field whereas an
odd differential form is the same as a totally antisymmetric covariant
pseudotensor field.) 

We assume that the electric current $J_{abc}$ and the magnetic 
current $M_{abc}$ are given by equations of the form
\begin{equation}\label{eq:currents}
  J_{abc} (x) = j_{abc} \big( x, F(x), H(x) \big) \; , \qquad
  M_{abc} (x) = m_{abc} \big( x, F(x), H(x) \big) \; .
\end{equation}
Here it is essential that the values of the currents at $x$ depend 
on $x$ and on the values of field strength and excitation
at $x$, but not on their derivatives. We will soon see that, 
under this assumption, the currents are irrelevant for the 
question of whether the initial value problem is well-posed.
Therefore, they will play no role in our further discussion;
we have allowed for non-zero currents only for the sake of
generality. Note, however, that the functions $j_{abc}$ and
$m_{abc}$ in (\ref{eq:currents}) are not completely arbitrary.
They must be consistent with the conservation laws
\begin{equation}\label{eq:jmcons}
  \partial _{[a} J_{bcd]} (x) \, = \, 0 \; , \qquad
  \partial _{[a} M_{bcd]} (x) \, = \, 0  \; ,
\end{equation}
which are a consequence of (\ref{eq:maxwell}).

In addition to (\ref{eq:currents}) we assume that we have a 
constitutive law in the form of six scalar equations
\begin{equation}\label{eq:constitutive}
  \mathcal{F}_A \big( x, F(x) , H(x) \big) = 0 \; , 
\qquad A=1, \dots ,6 
\end{equation}
which allow to express six of the twelve independent components 
$F_{ab}(x)$ and $H_{ab}(x)$ in terms of the remaining six. (We 
shall later specify the six components which are to be eliminated.) 
Again, it is essential that the constitutive law is local in the 
sense that knowledge of six components at a particular point $x$ 
allows to express the remaining six components at this particular 
point $x$. In particular, it is essential that (\ref{eq:constitutive}) 
does not involve derivatives of the field components.

We now separate the 8 equations  (\ref{eq:maxwell}) into two 
constraints
\begin{equation}\label{eq:constraints}  
  \partial _{[\mu} F_{\nu \sigma]} (x) = J_{\mu \nu \sigma} (x) \; ,
  \qquad
  \partial _{[\mu} H_{\nu \sigma]} (x) = M_{\mu \nu \sigma} (x) \; ,
\end{equation}
which do not contain any $\partial _0$ derivative, and six evolution
equations
\begin{equation}\label{eq:evolution}  
  \partial _{[0} F_{\nu \sigma]} (x) = J_{0 \nu \sigma} (x) \; ,
  \qquad
  \partial _{[0} H_{\nu \sigma]} (x) = M_{0 \nu \sigma} (x) \; ,
\end{equation}
which do contain $\partial _0$ derivatives. For the well-posedness of
the initial-value problem, only the evolution equations are relevant.
The constraints restrict the allowed initial values. After solving
the evolution equations with initial values that satisfy the constraints,
one has to check whether the constraints are preserved. This is guaranteed
if the currents (\ref{eq:currents}) satisfy the conservation laws
(\ref{eq:jmcons}).

To link up with standard notation of electrodynamics, we decompose
field strength and excitation in electric and magnetic parts.
\begin{equation}\label{eq:EBHD}
\begin{split}
  F_{\nu 0} = E_{\nu} \; , \qquad
  F_{\mu \rho} = \epsilon _{\mu \rho \sigma} B^{\sigma} \; , \quad
\\
  H_{\nu 0} = - \mathcal{H}_{\nu} \; , \qquad
  H_{\mu \rho} = \epsilon _{\mu \rho \sigma} D^{\sigma} \; , 
\end{split}
\end{equation}
where $\epsilon _{\mu \rho \sigma}$ is the 3-dimensional Levi-Civita symbol,
defined by the properties that it is totally antisymmetric and satisfies 
$\epsilon _{123} =1$. 

This puts the constraints (\ref{eq:constraints}) into the form
\begin{equation}\label{eq:constraintsDB}
\partial _{\rho} D^{\rho} \, + \, \dots \, = \, 0 \; , \qquad
\partial _{\rho} B^{\rho} \, + \, \dots \, = \, 0 \; , 
\end{equation}
where the ellipses indicate terms that do not involve derivatives of 
the fields. The evolution equations (\ref{eq:evolution}) can be
conveniently written in six-vector form. To that end, we write $\vec{E}$ 
and $\vec{\mathcal{H}}$ for the three-column vectors with components 
$E_1, E_2, E_3$ and $\mathcal{H}_1,\mathcal{H}_2,\mathcal{H}_3$, 
respectively, and we write $\vec{D}$ and $\vec{B}$ for the three-column 
vectors with components $D^1, D^2, D^3$ and $B^1, B^2, B^3$, respectively. 
Then the evolution equations (\ref{eq:evolution}) read
\begin{equation}\label{eq:six}
  \partial _0 \, 
  \begin{pmatrix}
     \vec{D} \\ \vec{B}
  \end{pmatrix}
  \, - \, 
  \begin{pmatrix} 
    \mathbf{0} &  - \mathbf{A}^{\rho} \\
    \mathbf{A}^{\rho} & \mathbf{0} 
  \end{pmatrix}
  \, \partial _{\rho} \,
  \begin{pmatrix}
     \vec{E} \\ \vec{\mathcal{H}}
  \end{pmatrix}
  \, + \, \dots \, = \, 0 \; .
\end{equation}
Here the $3 \times 3$ matrices $\mathbf{A}^{\rho}$ are defined by
\begin{equation}\label{eq:Arho}
  \mathbf{A}^1 \, = \, 
  \begin{pmatrix}
     0 & 0 & 0 \\
     0 & 0 & 1 \\
     0 & -1 & 0 
  \end{pmatrix}
  \; , \quad 
  \mathbf{A}^2 \, = \, 
  \begin{pmatrix}
     0 & 0 & -1 \\
     0 & 0 & 0 \\
     1 & 0 & 0 
  \end{pmatrix}
  \; , \quad 
  \mathbf{A}^3 \, = \, 
  \begin{pmatrix}
     0 & 1 & 0 \\
     -1 & 0 & 0 \\
     0 & 0 & 0 
  \end{pmatrix}
\end{equation}
and, as in (\ref{eq:constraintsDB}), the ellipses in (\ref{eq:six}) indicate terms 
that do not involve derivatives of the fields.

We shall now require that the constitutive equations (\ref{eq:constitutive})
can be solved for $\vec{E} (x)$ and $\vec{\mathcal{H}} (x)$. If this is the case, we
can eliminate $\vec{E}$ and $\vec{\mathcal{H}}$ from (\ref{eq:six}) with the help of
the constitutive equations, and we are left with a set of six first-order
differential equations for the six dynamical variables $\vec{B}$ and 
$\vec{D}$. If, on the other hand, the constitutive equations cannot be 
solved for $\vec{E} (x)$ and $\vec{\mathcal{H}} (x)$, the number of evolution 
equations does not coincide with the number of independent dynamical 
variables, so there is no chance to get a well-posed initial-value
problem. We introduce the following terminology.

\begin{defin}\label{def:regular}
A coordinate system is called \emph{admissible} if, in this coordinate system,
the constitutive law (\ref{eq:constitutive}) can be solved for $\vec{E} (x)$ 
and $\vec{\mathcal{H}} (x)$. A constitutive law is called \emph{regular} at $x$
if there is an admissible coordinate system on some neighborhood of $x$.
\end{defin}
Henceforth we assume that we have a constitutive law that is 
regular at some chosen point, and we work in an admissible 
coordinate system defined on some open neighbourhood $U$ 
of this point in $M$. We will see in Section \ref{sec:trafo} 
that then, if $U$ is chosen sufficiently small, almost all other 
coordinate systems on $U$ are admissible as well, and we will 
investigate the behaviour under coordinate changes of all
relevant quantities.

The assumption that (\ref{eq:constitutive}) can be solved for $\vec{E} (x)$ and 
$\vec{\mathcal{H}} (x)$ results in equations of the form
\begin{equation}\label{eq:impermittivity}
  \partial _{\rho} E_{\mu} (x) \, = \,  
  \kappa _{\mu \tau} \big( x , \vec{D}(x) , \vec{B}(x) \big)
  \, \partial _{\rho} D^{\tau} (x) \, + \, 
  \chi _{\mu \tau} \big( x , \vec{D}(x) , \vec{B}(x) \big)
  \, \partial _{\rho} B^{\tau} (x) \, + \, \dots 
\end{equation}
\begin{equation}\label{eq:impermeability}
  \partial _{\rho} \mathcal{H}_{\mu} (x) \, = \,  
  \gamma _{\mu \tau} \big( x , \vec{D}(x) , \vec{B}(x) \big)
  \, \partial _{\rho} D ^{\tau}(x) \, + \, 
  \nu _{\mu \tau} \big( x , \vec{D}(x) , \vec{B}(x) \big)
  \, \partial _{\rho} B ^{\tau} (x) \, + \, \dots
\end{equation}
Here, as before, the ellipses indicate terms that do not involve 
derivatives of the fields. In the more particular case that the
constitutive law is linear, the coefficients ${\kappa}_{\mu \tau}$, 
$\chi _{\mu\tau}$, $\gamma _{\mu \tau}$ and $\nu _{\mu \tau}$
depend only on $x$ but not on the fields. In the
following we denote by $\boldsymbol{\kappa}$, $\boldsymbol{\nu}$,
$\boldsymbol{\chi}$ and $\boldsymbol{\gamma}$ the $3 \times 3$ matrices 
with components $\kappa _{\mu \tau}$, $\nu _{\mu \tau}$, $\chi _{\mu \tau}$
and $\gamma _{\mu \tau}$, respectively. $\boldsymbol{\kappa}$ is 
called the \emph{impermittivity} matrix, $\boldsymbol{\nu}$ 
is called the \emph{impermeability} matrix, and $\boldsymbol{\chi}$ 
and $\boldsymbol{\gamma}$ are called the \emph{magneto-electric 
cross-terms}. (Our notation follows Kong \cite{Kong1974,Kong1975}.) 
The standard text-book formalism 
of electrodynamics is recovered if we assume that $\boldsymbol{\chi}$ and 
$\boldsymbol{\gamma}$ vanish and that $\boldsymbol{\kappa}$ and 
$\boldsymbol{\nu}$ depend only on $x$ and are invertible. Then 
$\boldsymbol{\varepsilon} = \boldsymbol{\kappa}^{-1}$ is called 
the \emph{permittivity} (or \emph{dielectricity}) matrix and 
$\boldsymbol{\mu} = \boldsymbol{\nu}^{-1}$ is called the 
\emph{permeability} matrix. A priori, however, there is no reason 
to assume that we can choose our coordinate system such that
the magneto-electric cross-terms vanish and that 
$\boldsymbol{\kappa}$ and $\boldsymbol{\nu}$ are invertible. For a 
detailed discussion of media with magneto-electric cross-terms see
O'Dell \cite{Odell1970}.

The four $3 \times 3$ matrices $\boldsymbol{\kappa}$, $\boldsymbol{\nu}$,
$\boldsymbol{\chi}$ and $\boldsymbol{\gamma}$ can be combined into the 
$6 \times 6$ matrix
\begin{equation}\label{eq:defM}
  \mathbf{M} =   
  \begin{pmatrix} 
    \boldsymbol{\kappa} & \boldsymbol{\chi} \\
    \boldsymbol{\gamma} & \boldsymbol{\nu} 
  \end{pmatrix}
\end{equation}
which we call the \emph{constitutive matrix}.

With (\ref{eq:impermittivity}) and (\ref{eq:impermeability}) inserted
into (\ref{eq:six}), we get the following set of six evolution equations
for the six dynamical variables $\vec{D}$ and $\vec{B}$.
\begin{equation}\label{eq:principal}
  \partial _0 \, 
  \begin{pmatrix}
     \vec{D} \\ \vec{B}
  \end{pmatrix}
  \, - \, \mathbf{L} ^{\rho}  \, 
  \partial _{\rho} \,
  \begin{pmatrix}
     \vec{D} \\ \vec{B}
  \end{pmatrix}
  \, + \, \dots \, = \, 0 \; ,
\end{equation}
where the $6 \times 6$ matrix 
\begin{equation}\label{eq:Lrho}
  \mathbf{L} ^{\rho} \; = \;
  \begin{pmatrix} 
    \mathbf{0} &  - \mathbf{A}^{\rho} \\
     \mathbf{A}^{\rho} & \mathbf{0} 
  \end{pmatrix}
  \, 
  \begin{pmatrix} 
    \boldsymbol{\kappa} & \boldsymbol{\chi} \\
    \boldsymbol{\gamma} & \boldsymbol{\nu} 
  \end{pmatrix}
\end{equation}
depends on $x$ and, in the case of a non-linear constitutive law, also
on $\vec{D}(x)$ and $\vec{B}(x)$. Thus, (\ref{eq:principal}) is a quasilinear
system of partial differential equations with non-constant coefficients. 

Before we proceed further it is useful to add a remark on the fact that we had 
to solve the constitutive equations for $\vec{E}$ and $\vec{\mathcal{H}}$,
rather than for any other combination of field components, as a necessary
condition for having a well-posed initial-value problem. It is sometimes 
argued (see, e.g., O'Dell \cite{Odell1970}, Section 2.1, or Hehl and 
Obukhov \cite{HehlObukhov2005}) that one should solve the constitutive 
equations either for $\vec{E}$ and $\vec{B}$, or for
$\vec{D}$ and $\vec{\mathcal{H}}$, because only then has the resulting
equation a covariant (i.e., four-dimensional, coordinate-independent)
meaning. According to this point of view, constitutive equations solved 
for other combinations of the field components are ``a historical artifact''
and ``should be phased out from use'' \cite{HehlObukhov2005}.
It is, indeed, true that the condition of solvability for 
$\vec{E}$ and $\vec{\mathcal{H}}$ is not covariant. Nonetheless,
it is precisely this condition which appears if we ask for a well-posed
initial-value problem. This should not come as a surprise. The initial-value
problem refers to a particular slicing of the spacetime into hypersurfaces
$x^0 = \mathrm{constant}$. It is \emph{not} a problem that has a covariant
answer; the initial-value problem is well-posed for some slicings, and
not well-posed for others. So it is quite natural that non-covariant
conditions play a role. 

If the constitutive law can be solved not only for $\vec{E}$ and 
$\vec{\mathcal{H}}$ but also for $\vec{E}$ and $\vec{B}$, the 
impermittivity matrix $\boldsymbol{\kappa}$ must be invertible.
The constitutive matrix (\ref{eq:defM}) can then be written in the form
\begin{equation}\label{eq:transM}
  \mathbf{M} \,
  \; = \;
  \begin{pmatrix} 
    \mathbf{1}  & \mathbf{0} \: 
    \\[0.1cm]
    \boldsymbol{\gamma} \boldsymbol{\kappa} ^{-1} 
    & \mathbf{1} \:
  \end{pmatrix}
  \begin{pmatrix} 
    \: \boldsymbol{\kappa} \: 
    &  \mathbf{0} 
    \\[0.1cm]
    \mathbf{0} & \: \boldsymbol{\nu} - 
    \boldsymbol{\gamma} \boldsymbol{\kappa} ^{-1} \boldsymbol{\chi} \:
  \end{pmatrix}
  \begin{pmatrix} 
    \: \mathbf{1} \: & \: \boldsymbol{\kappa} ^{-1} \boldsymbol{\chi} \:
    \\[0.1cm]
    \mathbf{0} & \mathbf{1}
  \end{pmatrix}
\end{equation}
as can be easily verified by multiplying out the right-hand side. 
(\ref{eq:transM}) implies
\begin{equation}\label{eq:detM}
\mathrm{det} ( \mathbf{M} ) \, = \, \mathrm{det} (\boldsymbol{\nu} -
\boldsymbol{\gamma} \boldsymbol{\kappa} ^{-1} \boldsymbol{\chi})
\, \mathrm{det} ( \boldsymbol{\kappa} ) \; .
\end{equation}
Analogously, if $\boldsymbol{\nu}$ is invertible, we find
\begin{equation}\label{eq:transM2}
  \mathbf{M} \,
  \; = \;
  \begin{pmatrix} 
    \: \mathbf{1} \: & \: \boldsymbol{\chi} \boldsymbol{\nu} ^{-1} \: 
    \\[0.1cm]
    \mathbf{0} & \mathbf{1}
  \end{pmatrix}
  \begin{pmatrix} 
    \: \boldsymbol{\kappa} - \boldsymbol{\chi} \boldsymbol{\nu} ^{-1} \boldsymbol{\gamma} \:
    &  \: \mathbf{0} \:
    \\[0.1cm]
    \mathbf{0} & \boldsymbol{\nu}  
  \end{pmatrix}
  \begin{pmatrix} 
    \mathbf{1} & \: \mathbf{0} \:
    \\[0.1cm]
    \: \boldsymbol{\nu} ^{-1} \boldsymbol{\gamma} \: & \mathbf{1}
  \end{pmatrix}
\end{equation}
and hence
\begin{equation}\label{eq:detM2}
\mathrm{det} ( \mathbf{M} ) \, = \, \mathrm{det} (\boldsymbol{\kappa} -
\boldsymbol{\chi} \boldsymbol{\nu} ^{-1} \boldsymbol{\gamma})
\, \mathrm{det} ( \boldsymbol{\nu} ) \; .
\end{equation}
Equations (\ref{eq:transM}) and (\ref{eq:transM2}) are useful for
calculating the inverse of $\mathbf{M} \,$.

\section{Approximate-plane-wave solutions of Maxwell's equations}
\label{sec:planewave}
\noindent
Now we want to derive the characteristic equation of the evolution 
equations (\ref{eq:principal}) and the resulting light cone
structure. There are two quite different methods of how to do this. The
first method, which goes back to Hadamard, investigates the
directions in which discontinuities of the electromagnetic field
can propagate. For linear constitutive laws on a bare manifold,
this method is used in the book by Hehl and Obukhov
\cite{HehlObukhov2003}. The second method investigates 
the directions in which approximate-plane-wave solutions of
Maxwell's equations can travel. This method was pioneered by 
Luneburg whose work is reviewed in the book by Kline and
Kay \cite{KlineKay1965}; their treatment is restricted to
linear and isotropic constitutive laws on Minkowski spacetime.

Here we want to use the second method because it provides
us with a clear physical interpretation of the characteristic
equation. As there are no treatments in the literature that
cover our situation -- Maxwell's equations with a local but 
possibly nonlinear constitutive law on a bare manifold --, we
give a detailed and self-contained exposition. Our first task
is to define the notion of an ``approximate-plane-wave solution''. 
 
In standard electrodynamics on Minkowski spacetime, wave propagation
can be studied in terms of plane harmonic waves. Maxwell's equations
on a bare manifold do not admit plane-harmonic-wave solutions in general. 
However, they do admit such solutions in an approximative sense. To make 
this mathematically precise, we introduce the following terminology.
\begin{defin}\label{def:wave}
An \emph{approximate-plane-wave family} with background field 
$\vec{D}(x)$ , $\vec{B}(x)$ is a one-parameter family
\begin{equation}\label{eq:wave}
  \begin{pmatrix}
     \vec{\mathfrak{D}} (\alpha , x) \\[0.2cm] \vec{\mathfrak{B}} (\alpha , x)
  \end{pmatrix}
   \, = \, 
  \begin{pmatrix}
     \vec{D} (x) \\[0.2cm] \vec{B} (x)
  \end{pmatrix}
  \, + \, \alpha \,  
  \mathrm{Re} \left\{ \, \Big(
  \begin{pmatrix}
     \vec{\, d} ( x) \\[0.2cm] \vec{\, b} ( x)
  \end{pmatrix}
   \, + \, O( \alpha ) \, \Big)  \, \mathrm{exp} \big(i S(x)/{\alpha} \big) 
   \; \right\}
\end{equation}
with the following properties:
\begin{itemize}
\item[(a)]
The coordinates $x = (x^0,x^1,x^2,x^3)$ range over an open neighborhood
$U$ of the manifold $M$ and the parameter $\alpha$ ranges over the
strictly positive real numbers, $\alpha \in \mathbb{R}^+$.
\item[(b)]
$\vec{D}$ and $\vec{B}$ are $\mathbb{R}^3$ valued $C^{\infty}$ functions. 
\item[(c)]
$S$ is a real-valued $C^{\infty}$ function whose gradient $dS(x) = 
\partial _a S(x) \, dx^a$ has no zeros on $U$. 
We refer to $S$ as to the \emph{eikonal function} of the 
approximate-plane-wave family.
\item[(d)]
$\vec{\, d}$ and $\vec{\, b}$ are $\mathbb{C}^3$ valued
$C^{\infty}$ functions with $ \big( \vec{\, d}(x) , \vec{\, b}(x) \big)
 \, \neq \, \big( \vec{\, 0} , \vec{\, 0} \big)$
for all $x$ in $U$,
\end{itemize}
\end{defin}
If $\partial _a S$, $\vec{\, d}$ and $\vec{\, b}$ are independent
of $x$ (in the chosen coordinate system) and the $O(\alpha )$ terms 
in (\ref{eq:wave}) are zero, (\ref{eq:wave}) is a background field 
with an $\alpha$-dependent plane harmonic wave added; the wave covector 
of this plane harmonic wave is given by $k_a = \partial _a S / \alpha$.
This observation gives the following interpretation to an arbitrary 
approximate-plane-wave family. On a sufficiently small neighborhood, 
$\partial _a S$, $\vec{\, d}$ and $\vec{\, b}$ differ arbitrarily 
little from constants, and for $\alpha$ sufficiently small the 
$O(\alpha )$ terms give arbitrarily small contributions. Thus, on a 
small neighborhood and for small $\alpha$, (\ref{eq:wave}) can be 
approximately viewed as a plane harmonic wave added to the background 
field. The smaller $\alpha$, the more oscillations we have on the 
chosen neighborhood. We refer to $\alpha \to 0$ as to the 
\emph{high-frequency limit}. 

For the evaluation of approximate-plane-wave families the following
simple lemma is crucial.
\begin{lem}\label{lem:lim}
Let $S$ be the eikonal function of an approximate-plane-wave
family and let $u$ be a complex-valued continuous function 
defined on the same neighborhood $U$ as $S$. Then  
$\underset{\alpha \, \to \, 0}{\mathrm{lim}} \, \mathrm{Re}
\big\{ u(x) e^{iS(x)/\alpha} \big\}$ exists for all $x$ in
$U$ if and only if $u(x)=0$ for all $x$ in $U$.
\end{lem}
\begin{proof}
If $u$ is different from zero at some point in $U$, it is 
different from zero on an open subset $V$ of $U$. For almost all
$x$ in $V$ we have $S(x) \neq 0$, as $dS$ has no zeros. As a 
consequence, for almost all $x$ in $V$ the limit does not exist.
\end{proof} 

We now consider an approximate-plane-wave family, given in 
coordinates $x=(x^0,x^1,x^2,x^3)$ that are admissible in the
sense of Definition \ref{def:regular}. 
We assume that the background fields $\vec{D} (x)$ and 
$\vec{B} (x)$ satisfy the constraints (\ref{eq:constraintsDB}) 
and the evolution equations (\ref{eq:principal}). We then say
that the approximate-plane-wave family (\ref{eq:wave}) is an 
$N^{\mathrm{th}}$ order \emph{asymptotic solution} to the constraints 
(\ref{eq:constraintsDB}) if
\begin{equation}\label{eq:asymptoticcon}
  \underset{\alpha \, \to \, 0}{\mathrm{lim}} \;
  \left\{ \, \frac{1}{\alpha ^N} \, 
  \big( 
  \, \partial _{\rho} \, \mathfrak{D}^{\rho} (\alpha , x) 
  \, + \, \dots \,
  \big) \; \right\}
  \, = \, 0 \; , \qquad
  \underset{\alpha \, \to \, 0}{\mathrm{lim}} \;
  \left\{ \, \frac{1}{\alpha ^N} \, 
  \big( 
  \, \partial _{\rho} \, \mathfrak{B}^{\rho} (\alpha , x)
  \, + \, \dots \,
  \big) \; \right\}
  \, = \, 0 \; , 
\end{equation}
and to the evolution equations (\ref{eq:principal}) if
\begin{equation}\label{eq:asymptotic}
\underset{\alpha \, \to \, 0}{\mathrm{lim}} \;
\left\{ \, \frac{1}{\alpha ^N} \, 
\left( 
  \partial _0 \, 
  \begin{pmatrix}
     \vec{\mathfrak{D}} (\alpha , x) 
  \\[0.2cm]
   \vec{\mathfrak{B}} (\alpha , x )
  \end{pmatrix}
  \, - \, 
  \mathbf{L} ^{\rho} \big( 
  x , \vec{\mathfrak{D}}(\alpha , x) , \vec{\mathfrak{B}}(\alpha , x) 
  \big) \, 
  \partial _{\rho} \,
  \begin{pmatrix}
     \vec{\mathfrak{D}} (\alpha , x) 
  \\[0.2cm]
   \vec{\mathfrak{B}} (\alpha , x )
  \end{pmatrix}
  \, + \, \dots \, 
\right) \, 
\right\} \; = \; 
  \begin{pmatrix}
     \vec{\, 0} \\ \vec{\, 0}
  \end{pmatrix}
  \; .
\end{equation}
Here the ellipses stand for the same terms as in (\ref{eq:constraintsDB})
and (\ref{eq:principal}), respectively. It is obvious that an $N^{\mathrm{th}}$ 
order asymptotic solution is automatically an $M^{\mathrm{th}}$ order 
asymptotic solution for all $M \le N$. We want to investigate the lowest 
non-trivial order $N=0$, which is known as the \emph{geometric optics
approximation}. (We will not consider asymptotic solutions of higher order, 
which allow to determine the $O(\alpha)$ terms in (\ref{eq:wave}) iteratively.) 
Then in (\ref{eq:asymptoticcon}) and (\ref{eq:asymptotic}) 
the terms indicated by ellipses, which do not contain derivatives of the 
fields, give no contribution to the limit. Using Lemma \ref{lem:lim}, we 
find that an approximate-plane-wave family (\ref{eq:wave}) is an asymptotic 
solution of order $N=0$ to the 
constraints if and only if 
\begin{equation}\label{eq:constraintsS}
  \partial _{\rho} S (x) \, d^{\rho} (x) \, = \, 0 \; , \qquad
  \partial _{\rho} S (x) \, b^{\rho} (x) \, = \, 0 \; , 
\end{equation}
and to the evolution equations if and only if
\begin{equation}\label{eq:eigenS}
  \partial _{\rho} S (x) \,
  \mathbf{L} ^{\rho}  \big(  x , \vec{D} (x) , \vec{B}(x) \big) \,  
  \begin{pmatrix}
     \vec{\, d} (x) \\ \vec{\, b} (x)
  \end{pmatrix}
  \; = \; 
  \partial _0 S (x) \, 
  \begin{pmatrix}
     \vec{\, d} (x) \\ \vec{\, b} (x)
  \end{pmatrix}
  \; .
\end{equation}
As the amplitudes are assumed to be non-zero, (\ref{eq:eigenS})
means that $\partial _0 S (x) $ must be an eigenvalue of the 
$6 \times 6$ matrix $\partial _{\rho} S (x) \, \mathbf{L} ^{\rho} 
\big( x , \vec{D} (x) , \vec{B}(x) \big)$, i.e.
\begin{equation}\label{eq:eikonal}
\mathrm{det} \, 
\left( 
\partial _0 S (x) \, \mathbf{1} \, - \, \partial _{\rho} S (x)\,
\mathbf{L} ^{\rho} \big(  x, \vec{D} (x) , \vec{B}(x) \big) \, 
\right)
\; = \; 0 \; .
\end{equation}
(\ref{eq:eikonal}) is the \emph{eikonal equation} of the evolution equations. 
It is a first order partial differential equation for $S$. Its
coefficients depend not only on $x$ but also on the background fields 
$\vec{D}(x)$ and $\vec{B}(x)$, unless we restrict to linear constitutive 
laws.

\section{The characteristic equation}
\label{sec:characteristic}
\noindent
Now we choose values for $x$, $\vec{D}(x)$ and $\vec{B}(x)$ and keep them fixed. 
We want to investigate the condition that there is a nonzero $\big( \vec{\, d} , 
\vec{\, b} \, \big)$ such that (\ref{eq:constraintsS}) and (\ref{eq:eigenS}) hold.
This is an algebraic condition on the covector $dS = \partial _a S \, dx^a$, i.e.,
it determines a subset of the cotangent space at $x$. We will call any covector 
which satisfies this condition a ``characteristic covector''. Hence, a covector 
is characteristic if and only if it is the gradient of an eikonal function of 
an approximate-plane-wave solution of order $N=0$. The precise definition reads
as follows. (As $x$, $\vec{D}(x)$, $\vec{B}(x)$ are now kept fixed, 
dependence on these quantities is no longer made explicit to ease notation.)
\begin{defin}\label{def:characteristic}
A covector $p_a dx^a$ is called \emph{characteristic} if there is $\big( \vec{\, d}
, \vec{\, b} \, \big) \neq \big( \vec{\, 0}, \vec{\, 0} \big)$ in $\mathbb{R}^6$ such that
\begin{equation}\label{eq:constraintsp}
  p _{\rho} \, d^{\rho} \, = \, 0 \; , \qquad
  p _{\rho} \, b^{\rho} \, = \, 0 \; , 
\end{equation}
\begin{equation}\label{eq:eigenp}
  p _{\rho} \,
  \mathbf{L} ^{\rho}  \, 
  \begin{pmatrix}
      \vec{\, d} \, \\  \vec{\, b} \,
  \end{pmatrix}
  \; = \; 
  p _0 \, 
  \begin{pmatrix}
      \vec{\, d} \, \\  \vec{\, b} \,
  \end{pmatrix}
  \; .
\end{equation}
The set of all characteristic covectors is called the \emph{characteristic
variety}. 
\end{defin}
We will show in Section \ref{sec:trafo} that this definition is 
coordinate-independent. Clearly, a necessary condition for 
$p_a dx^a$ to be characteristic is that it satisfies 
\begin{equation}\label{eq:characteristic}
\mathrm{det}  
\left( 
p _0 \, \mathbf{1} \, - \, p _{\rho} \,
\mathbf{L} ^{\rho} \, 
\right)
\; = \; 0 \; .
\end{equation}
(\ref{eq:characteristic}) is called the \emph{characteristic equation}
and its left-hand side is called the \emph{characteristic polynomial}
of the evolution equations. Note that (\ref{eq:characteristic}) makes
sure that (\ref{eq:eigenp}) admits a non-trivial solution but does
not take (\ref{eq:constraintsp}) into account.

By writing covectors as $p_a dx^a$ we have introduced canonical momentum
coordinates $p_a$ in each cotangent space which are conjugate to our
admissible coordinates $x^a$. If $x^0$ can be interpreted as a temporal
coordinate and the $x^{\mu}$ as spatial coordinates, we refer to $p_0$
as to the \emph{frequency} and to $\vec{p}$ as to the \emph{spatial wave 
covector} (for dimensional reasons, one may put in a factor $\hbar$); 
here and in the following, $\vec{p}$ stands for the three-column vector 
with components $p_1, p_2, p_3$. 

The left-hand side of (\ref{eq:characteristic}) is a sixth order
homogeneous polynomial in the variables $p_a$. For each $\vec{p} \in 
\mathbb{R}^3$, it is a sixth order polynomial in the variable $p_0$. 
As such, it has six complex roots which will be denoted by 
$\omega _A ( \vec{p} )$ where $A=1, \dots , 6$. In this notation, 
the characteristic equation (\ref{eq:characteristic}) takes the form 
\begin{equation}\label{eq:omega}
    \prod _{A=1} ^6 \big( p_0 - \omega _A ( \vec{p} ) \big) 
  \: = \: 0 \: .
\end{equation}
The $\omega _A ( \vec{p} )$ are the eigenvalues of the 
matrix $p_{\rho} \mathbf{L}^{\rho}$, 
\begin{equation}\label{eq:eigen}
  p_{\rho} \mathbf{L}^{\rho} \, 
  \begin{pmatrix} 
    \vec{\, d}_A ( \vec{p} ) \\ \vec{\, b}_A ( \vec{p} ) 
  \end{pmatrix}
  \, = \, 
  \omega _A ( \vec{p})  \, 
  \begin{pmatrix} 
    \vec{\, d}_A ( \vec{p} ) \\ \vec{\, b}_A ( \vec{p} ) 
  \end{pmatrix} 
  \; .
\end{equation}
In general they are complex, and so are the components of the eigenvectors
$\vec{\, d}{}_A (\vec{p})$ and $\vec{\, b}{}_A ( \vec{p})$.
Only real $\omega _A ( \vec{p} )$ are related to 
approximate-plane-wave solutions, because the eikonal function 
is supposed to be a real function. (Non-real eigenvalues are associated 
not with oscillating modes but rather with so-called \emph{evanescent 
modes}, i.e., with exponentially decaying fields.) From the real
solutions $\omega_A (\vec{p})$ of the eigenvalue problem 
(\ref{eq:eigen}) we have to single out those for which the eigenvectors
satisfy the constraints (\ref{eq:constraintsp}); this will give us all
the characteristic covectors in the form $\omega _A (\vec{p} ) dx^0
+ p_{\mu} dx^{\mu}$.

If $\vec{p}$ runs over $\mathbb{R} ^3 \setminus \{\vec{\, 0}\}$, $p_{\rho}
\mathbf{A} ^{\rho}$ runs over all non-zero antisymmetric $3 \times 3$ matrices.
Any such matrix has a one-dimensional kernel. As we can read from 
(\ref{eq:Lrho}), this implies that the kernel of the matrix $p_{\rho} 
\mathbf{L}^{\rho}$ must be at least two-dimensional, hence $p_0 = 0$
is an eigenvalue of algebraic multiplicity $\ge 2$ and the characteristic 
equation is of the form
\begin{equation}\label{eq:linfac}
  p_0 ^2 \prod _{A=1} ^4 \big( p_0 - \omega _A (\vec{p} ) \big) 
  \; = \, 0 \; . 
\end{equation}
The eigenvalues $\omega _A ( \vec{p} )$, for $A=1, \dots ,4$, are in general
complex and some of them may be zero as the algebraic multiplicity of the
eigenvalue $p_0=0$ may be bigger than 2.
If $\mathbf{M}$ is invertible, the kernel of the matrix $p_{\rho} 
\mathbf{L}^{\rho}$ is precisely two-dimensional, hence $p_0=0$ is 
an eigenvalue of geometric multiplicity equal to 2; however, even 
in this case the algebraic multiplicity may be bigger than 2.

We will now show that, if $\omega _A ( \vec{p} )\neq 0$, the 
corresponding eigenvectors automatically satisfy the constraints 
(\ref{eq:constraintsp}). 

\begin{proposition}\label{prop:transverse}
  If (\ref{eq:eigen}) holds with $\omega _A ( \vec{p} ) \neq 0$, the  
  $\vec{\, d}_A (\vec{p} )$ and 
  $\vec{\, b}_A (\vec{p} )$ satisfy $d^{\rho} _A (\vec{p} ) p_{\rho} =  
  b^{\rho} _A (\vec{p} ) p_{\rho} = 0$.
\end{proposition}
\begin{proof}
  The eigenvalue equation (\ref{eq:eigen}) implies that, for all complex 
  numbers $a$ and $b$,
\begin{equation}\label{eq:mult}
  \begin{pmatrix} 
    a \, \vec{p} \\ b \, \vec{p} 
  \end{pmatrix}
  \, \cdot \,
    \begin{pmatrix} 
    \mathbf{0} &  - p_{\rho} \mathbf{A}^{\rho} \\
    p_{\rho} \mathbf{A}^{\rho} & \mathbf{0} 
  \end{pmatrix}
  \, 
  \begin{pmatrix} 
    \boldsymbol{\kappa} & \boldsymbol{\chi} \\
    \boldsymbol{\gamma} & \boldsymbol{\nu} 
  \end{pmatrix}
  \,  
  \begin{pmatrix} 
    \vec{\, d}_A ( \vec{p} ) \\ \vec{\, b}_A ( \vec{p} ) 
  \end{pmatrix}
  \, = \, 
  \begin{pmatrix} 
    a \, \vec{p} \\ b \, \vec{p} 
  \end{pmatrix}
  \, \cdot \,
  \omega _A ( \vec{p})  \, 
  \begin{pmatrix} 
    \vec{\, d}_A ( \vec{p} ) \\ \vec{\, b}_A ( \vec{p} ) 
  \end{pmatrix} 
  \; ,
\end{equation}
where the dot denotes the standard scalar product in $\mathbb{C}^6$. Using the antisymmetry
of $p_{\rho} \mathbf{A} ^{\rho}$, this can be rewritten as 
\begin{equation}\label{eq:mult2}
    \begin{pmatrix} 
    \mathbf{0} &  - p_{\rho} \mathbf{A}^{\rho} \\
    p_{\rho} \mathbf{A}^{\rho} & \mathbf{0} 
  \end{pmatrix}
  \, 
  \begin{pmatrix} 
    a \, \vec{p} \\ b \, \vec{p} 
  \end{pmatrix}
  \, \cdot \,
  \begin{pmatrix} 
    \boldsymbol{\kappa} & \boldsymbol{\chi} \\
    \boldsymbol{\gamma} & \boldsymbol{\nu} 
  \end{pmatrix}
  \,  
  \begin{pmatrix} 
    \vec{\, d}_A ( \vec{p} ) \\ \vec{\, b}_A ( \vec{p} ) 
  \end{pmatrix}
  \, = \, 
  \begin{pmatrix} 
    a \, \vec{p} \\ b \, \vec{p} 
  \end{pmatrix}
  \, \cdot \,
  \omega _A ( \vec{p})  \, 
  \begin{pmatrix} 
    \vec{\, d}_A ( \vec{p} ) \\ \vec{\, b}_A ( \vec{p} ) 
  \end{pmatrix} 
  \; .
\end{equation}
As $\vec{p}$ spans the kernel of $p_{\rho} \mathbf{A} ^{\rho}$, the left-hand side
vanishes, so the right-hand side has to vanish for all complex numbers $a$ and $b$.
As $\omega _A (\vec{p} )$ is non-zero, this completes the proof.
\end{proof}
One may interpret the constraints (\ref{eq:constraintsp}) as saying that 
$\vec{\, d}$ and $\vec{\, b}$ must be \emph{transverse}. By Proposition 
\ref{prop:transverse}, this transversality condition follows already from 
the evolution equations for modes with non-zero frequency. It is interesting to 
note that such a transversality condition does not hold for $\vec{\, e}$ and
$\vec{\, h}$; on a manifold without a metric, transversality cannot even be 
formulated for these fields because one would have to raise an index. This is 
another indication that the choice of $\vec{D}$ and $\vec{B}$ as the dynamical 
variables is the natural one, cf. Kong \cite{Kong1974} and \cite{Kong1975}, 
Section 3.3, where an analogue of Proposition \ref{prop:transverse} is discussed 
for the special case of homogeneous linear media on Minkowski spacetime.

If the algebraic multiplicity of zero, as an eigenvalue of $p_{\rho} \mathbf{L} 
^{\rho}$, is equal to 2, Proposition \ref{prop:transverse} guarantees that a real 
covector $p_a dx^a$ is a characteristic covector if and only if it satisfies the 
\emph{reduced characteristic equation}
\begin{equation}\label{eq:redchar}
\prod _{A=1} ^4 \, 
\big( \, p_0 \, - \, \omega _A ( \vec{p} ) \, \big)
\, = \, 
\frac{1}{p_0^{\, 2}} \, \mathrm{det} \left( 
p _0 \, \mathbf{1} \, - \, p _{\rho} \,
\mathbf{L} ^{\rho} \, 
\right)
\; = \; 0 \; .
\end{equation}
Whereas (\ref{eq:characteristic}) is the ``characteristic equation of the 
evolution equations'', (\ref{eq:redchar}) can then be properly called the
``characteristic equation of the full Maxwell equations'' (evolution
equations plus constraints). 

The situation is very much more inconvenient if, for some $\vec{p} \in 
\mathbb{R} ^3 \setminus \{ \vec{0} \}$, the algebraic multiplicity of zero, 
as an eigenvalue of $p_{\rho} \mathbf{L} ^{\rho}$, is $>2$. In this case
one of the four roots $\omega _A ( \vec{p})$ of the reduced characteristic 
equation (\ref{eq:redchar}) is zero. The corresponding covector $\omega _a
( \vec{p} ) dx^0 + p_{\mu} dx^{\mu}$ may or may not be characteristic, i.e.,
the eigenvectors may or may not satisfy the constraints. If they do, 
we have an approximate-plane-wave solution of Maxwell's equations of order
$N=0$ with zero frequency. One would conclude from this observation
that $x^0$ cannot be interpreted as a temporal coordinate. Therefore 
it seems reasonable to discard such cases as unphysical, i.e., to restrict
to cases where the coordinates can be chosen such that zero-frequency
modes do not occur.

If $\boldsymbol{\kappa}$ is invertible, it is convenient to write 
the reduced characteristic polynomial (\ref{eq:redchar}) in the form
\begin{equation}\label{eq:TR}
 \prod _{A=1} ^4 \big( p_0 - \omega _A (\vec{p} ) \big)
 \; = \;
 \mathrm{det} ( \boldsymbol{\kappa} ) \,
 \mathcal{G}^{abcd}p_ap_bp_cp_d 
\end{equation}  
so that the characteristic equation reads $\mathcal{G}^{abcd}
p_ap_bp_cp_d = 0$. The reason for introducing the factor 
$\mathrm{det} ( \boldsymbol{\kappa} )$
is that the coefficients $\mathcal{G}^{abcd}$ transform like a 
tensor density, as will be proven in Section \ref{sec:trafo}.
Note, however, that in the case of a non-linear constitutive law the
$\mathcal{G}^{abcd}$ depend not only on $x$ but also on $\vec{D}(x)$
and $\vec{B}(x)$. Following Hehl and Obukhov \cite{HehlObukhov2003}
we call $\mathcal{G}^{abcd}$ the \emph{Tamm-Rubilar} tensor density. 
For an interesting representation of the Tamm-Rubilar tensor density,
using the adjugate (or classical adjoint) of a matrix, see Itin \cite{Itin2009}.

Before discussing the (reduced) characteristic equation in more detail, we add a 
word on terminology. We have used the term ``characteristic equation''
which is the standard notation in texts on partial differential equation. In
physics texts one finds the alternative term \emph{dispersion relation}.
Some authors, e.g. Hehl and Obukhov \cite{HehlObukhov2003}, use the term
\emph{Fresnel equation} as another alternative. Traditionally, the term 
``Fresnel equation'' is used in crystal optics for an equation that
determines the index of refraction (or, equivalently, the phase velocity)
in dependence of the spatial direction, see e.g. Born and Wolf
\cite{BornWolf2002}, Sect. 15.2.2. This equation is, indeed, equivalent
to the characteristic equation. 

\section{Hyperbolicity of the evolution equations}\label{sec:hyperbolicity}
\noindent
In the preceding section we have seen that the characteristic equation 
(\ref{eq:characteristic}), together with the constraints, determines the 
directions $p_adx^a$ in the cotangent space in which wavelike solutions 
can propagate. We will now discuss that the characteristic equation 
(\ref{eq:characteristic}) also contains all information that is necessary 
to decide whether or not the evolution equations determine a well-posed 
initial-value problem.

Recall that we denote by $\omega _A (\vec{p})$, for $A=1, \dots, 6$, the
six eigenvalues of the matrix $p_{\rho} \mathbf{L} ^{\rho}$ which depends
on $x$, $\vec{D}(x)$ and $\vec{B}(x)$. Two of these eigenvalues are zero,
$\omega _5 = \omega _6 = 0$, the other four are, in general, complex.
We now recall some standard terminology from the theory of partial differential
equations. 

The evolution equations are called \emph{hyperbolic} if all eigenvalues
$\omega _A \big( \vec{p} \big)$ are real, for all $\vec{p}$ in $\mathbb{R}^3$.
(This is the case if and only if, in the terminology of G{\aa}rding 
\cite{Garding1959}, the characteristic polynomial $\mathrm{det} 
\big( p_0 \mathbf{1} - p_{\rho} \mathbf{L}^{\rho} \big)$ is hyperbolic 
with respect to the covector $dx^0$.) The evolution equations are
called \emph{strongly hyperbolic} if for each $\vec{p}$ there is an 
invertible matrix $\mathbf{S} ( \vec{p} )$ such that $\mathbf{S} 
(\vec{p} ) ^{-1} p _ {\rho} \mathbf{L} ^{\rho} \mathbf{S} ( \vec{p} )$ 
is symmetric. They are called \emph{symmetric hyperbolic} if this is 
true with an $\mathbf{S}$ that is independent of $\vec{p}$. 
Obviously the following implications hold: symmetric hyperbolic $\Rightarrow$ 
strongly hyperbolic $\Rightarrow$ hyperbolic.

Of course, the answer to the question of whether any of the three 
properties -- hyperbolic, strongly hyperbolic or symmetric hyperbolic 
 -- holds, may vary in dependence of $x$, 
$\vec{D}(x)$ and $\vec{B}(x)$.

Hyperbolicity guarantees the unrestricted existence of 
approximate-plane-wave solutions. However, it is too weak to
guarantee well-posedness of the initial-value problem for the
class of differential equations -- quasilinear
with non-constant coefficients -- to which our 
evolution equations belong.

The latter requires strong hyperbolicity. More precisely, if the
evolution equations are strongly hyperbolic 
at some $x$, $\vec{D}(x)$ and $\vec{B}(x)$, the
following holds true. Data for $\vec{D}$ and $\vec{B}$ on the
hypersurface $x^0= \mathrm{constant}$ that take the prescribed
values $\vec{D}(x)$ and $\vec{B} (x)$ at $x$ determine a 
unique solution to the evolution equations on some neighborhood
of $x$, and the solution depends on the data continuously.
The data must be of Sobolev class $H^s$, for some $s \ge 3$, and
continuity is meant with respect to the Sobolev $H^s$ norm. For 
details and proofs the reader is referred to Taylor \cite{Taylor1991},
Theorem 5.2D. (Note that Taylor uses the term ``symmetrizable'' 
instead of ``strongly hyperbolic''.) 

In the more special case of symmetric hyperbolicity we have not 
only continuous dependence of the solution on the data but in
addition we can control the growth of the solution in
terms of \emph{energy inequalities}.  Usually symmetric
hyperbolicity is easier to check than hyperbolicity or
strong hyperbolicity. In Section \ref{sec:symhyp} below we
will give a convenient characterisation of all constitutive laws that
give symmetric hyperbolic evolution equations. For hyperbolic or
strongly hyperbolic evolution equations, no such characterisation
is known so far.

Lindell, Sihvola and Suchy \cite{LindellSihvolaSuchy1995}, in an 
otherwise very useful article, claim
that the eigenvalues $\omega _A (\vec{p})$ are real whenever the 
constitutive matrix $\mathbf{M}$ is symmetric. This would give a very
convenient sufficient condition for hyperbolicity. Unfortunately,
the claim is wrong; a counter-example is
\begin{equation}\label{eq:exM}
  \mathbf{M} =   
  \begin{pmatrix} 
    \boldsymbol{0} & \boldsymbol{1} \\
    \boldsymbol{1} & \boldsymbol{0} 
  \end{pmatrix}
\end{equation}
for which the matrix $p_{\rho} \mathbf{L}{}^{\rho}$ has  
eigenvectors $\omega _1 (\vec{p}) = \omega _2 (\vec{p}) = i | \vec{p}|$,
$\omega _3 (\vec{p}) = \omega _4 (\vec{p}) = -i | \vec{p}|$ and
$\omega _5 (\vec{p}) = \omega _6 (\vec{p}) = 0$. 
The error comes in eq. (65) of  \cite{LindellSihvolaSuchy1995} where
the authors divide by a quantity without paying attention to the fact
that this quantity may be zero. (Note that Lindell, Sihvola and Suchy 
allow the constitutive matrix to be complex, in contrast to the formalism
used here, and that their $\mathbf{M}$ is our $\mathbf{M}{}^{-1}$. 
In the above argument we specified their reasoning to the case that 
$\mathbf{M}$ is real and we used the obvious fact that an invertible real
matrix is symmetric if and only if its inverse is symmetric.)



\section{The light cones}\label{sec:cone}
\noindent
The homogeneity of the characteristic polynomial implies that the functions
$\omega _A$ are positively homogeneous of degree one, 
\begin{equation}\label{eq:poshom}
  \omega _A ( s \vec{p} ) \, = \, s \, \omega _A ( \vec{p} ) \qquad
  \text{for all} \: s > 0 \; .
\end{equation}
We now assume hyperbolicity, and we order the 4 roots of the reduced
characteristic polynomial according to
\begin{equation}\label{eq:order}
  \omega _1 ( \vec{p} )
  \ge \omega _2 ( \vec{p} )
  \ge \omega _3 ( \vec{p} )
  \ge \omega _4 ( \vec{p} )  \; .
\end{equation}
This guarantees that the $\omega _A$ are continuous, but not necessarily 
smooth, functions. Thus, the reduced characteristic equation defines four
connected sets 
\begin{equation}\label{eq:cone}
  \mathcal{C}_A \, = \, \{\; p_a \, dx^a \, | 
  \, p_0 = \omega _A ( \vec{p} ) \; \}\: ,
  \quad A=1, 2, 3, 4 \, ,
\end{equation}
in the (real) cotangent space of the chosen point $x$. Without our assumption
of hyperbolicity one would have to consider the $\mathcal{C}_A $ as subsets
of the complexified cotangent space. 

By (\ref{eq:poshom}), each $\mathcal{C} _A$ is a cone in the sense that it 
is invariant under multiplication with positive real numbers. We refer to 
the $\mathcal{C} _A$ as to the \emph{four branches of the characteristic 
variety} or, shorter, as to the four \emph{light cones}. By differentiating 
the characteristic equation we find immediately that the differential of
$p_0 - \omega _A ( \vec{p} )$ is non-zero at any point of $\mathcal{C}_A$
where $\mathcal{C}_A$ does not meet one of the other light cones.
Thus, $\mathcal{C}_A$ is a 3-dimensional manifold at any such point. At 
an intersection point with some other light cone $\mathcal{C}_B$, however, 
$\mathcal{C}_A$ need not be smooth. For instance, $\mathcal{C}_A$ and 
$\mathcal{C}_B$ may form a ``conical singularity'', with pointed
tips meeting head-on. This gives rise to the phenomenon of ``conical 
refraction'' whose observability has been a matter of vivid debate in 
the history of anisotropic optics; for a detailed discussion see e.g.
Born and Wolf \cite{BornWolf2002}, p. 813--818.

Our ordering (\ref{eq:order}) implies that 
$\mathcal{C}_4$ is the image of $\mathcal{C}_1$ and $\mathcal{C}_3$ is 
the image of $\mathcal{C}_2$ under reflection at the origin in the 
cotangent space. This follows immediately from the fact that the 
characteristic polynomial is homogeneous, Thus, $\mathcal{C}_1 \cup
\mathcal{C}_4$ is a double-cone in the sense that it is generated
by straight lines through the origin, and so is $\mathcal{C}_2
\cup \mathcal{C}_3$. Note, however, that in general there is 
no reflection symmetry with respect to the plane $p_0 = 0$. 

We will now show that the four light cones divide up into two past 
cones and two future cones, provided that the reduced characteristic
polynomial has no zero roots. In the terminology explained in Section 
\ref{sec:characteristic}, the latter condition means that we prohibit 
zero-frequency modes.

\begin{proposition}\label{prop:futpast}
Assume that the four roots of the reduced characteristic polynomial
are real and non-zero for all $\vec{p} \neq \vec{\, 0}$, and that we
order them according to (\ref{eq:order}). Then
\begin{equation}\label{eq:orderstrong}
 \omega _1 ( \vec{p} )  \ge \omega _2 ( \vec{p} )
 > 0 >  \omega _3 ( \vec{p} ) \ge \omega _4 ( \vec{p} )
\end{equation}
holds for all $\vec{p} \neq \vec{\, 0}$.
\end{proposition}
\begin{proof}
Assume that three of the four roots $\omega _A$ are positive 
at some $\vec{p}$. Owing to the homogeneity of the characteristic polynomial, 
three roots must be negative at $- \vec{p}$. However, as the $\omega _A$ are
continuous on $\mathbb{R}^3 \setminus \{ \vec{\, 0} \}$, this is possible
only if some $\omega _A $ has a zero somewhere on $\mathbb{R}^3 \setminus 
\{ \vec{\, 0} \}$ which contradicts our assumption. We have thus proven that 
it is impossible that three roots are positive. By the same token, it is 
impossible that three roots are negative, so we must have two positive and 
two negative roots, i.e., (\ref{eq:orderstrong}) must be true. 
\end{proof} 

From this proposition we find the following Corollary.

\begin{proposition}\label{prop:futpast2}
Assume that the evolution equations are strongly hyperbolic and that 
the constitutive matrix $\mathbf{M}$ is invertible. Then (\ref{eq:orderstrong})
holds for all $\vec{p} \neq \vec{\, 0}$.
\end{proposition}
\begin{proof}
What we have to prove is that our assumption prohibits zero-frequency modes,
i.e., that $p_0=0$ has algebraic multiplicity equal to 2 as an eigenvalue
of $p_{\rho} \mathbf{L} ^{\rho}$ for all $\vec{p} \neq \vec{0}$.  We have
aleady observed in Section \ref{sec:characteristic} that invertiblity of 
$\mathbf{M}$ guarantees that the geometric multiplicity of $p_0=0$ is 
equal to 2. Strong hyperbolicity makes sure that algebraic and geometric
multiplicity coincide. 
\end{proof} 

Whenever we have hyperbolicity, $\mathcal{C}_1$ and $\mathcal{C}_4$
are the boundaries of \emph{convex} open cones
\begin{equation}\label{eq:defZ}
  \mathcal{Z}_1 \, = \, \{\; p_a \, dx^a \; | 
  \; p_0 > \omega _1 ( \vec{p} ) \; \}\: , \quad
  \mathcal{Z}_4 \, = \, \{\; p_a \, dx^a \; | 
  \; p_0 < \omega _4 ( \vec{p} ) \; \}\: .
\end{equation}
The convexity of $\mathcal{Z}_1$ and $\mathcal{Z}_4$ is a general feature,
following from hyperbolicity, as was already proven in G{\aa}rding's
pioneering paper \cite{Garding1959}. By contrast, $\mathcal{C}_2$ and 
$\mathcal{C}_3$ are not in general the boundaries of convex sets.
One may call covectors in $\mathcal{Z}_1$ 
``future-pointing timelike'' and covectors in $\mathcal{Z}_4$ 
``past-pointing timelike''. Of course, future and past interchange 
their roles under reflection of the coordinate $x^0$. The fact 
that the light cones, as subsets of the cotangent space, are independent
of the chosen coordinate system will be proven in Section \ref{sec:trafo}.   



\section{Coordinate transformations}\label{sec:trafo}
\noindent
We have worked in a chosen admissible coordinate system throughout, and we will 
now investigate to what extent the results found are invariant with respect 
to coordinate transformations. In particular, we will verify that the light
cones are invariant (i.e., coordinate-independent).

As the hypersurfaces $x^0 = \mathrm{constant}$ play a distinguished role,
it is useful to distinguish coordinate transformations that leave these
hypersurfaces invariant. The most general such transformation induces on 
each cotangent space a linear transformation of the form
\begin{equation}\label{eq:Galileo1}
  d \tilde{x}{}^{\mu} = 
  a ^{\mu}{}_{\rho} \big( dx^{\rho} + v^{\rho} dx^0 \big)
  \: , \qquad
  d \tilde{x}{}^0 = c \, dx ^0 
\end{equation}
and on each tangent space the dual linear transformation
\begin{equation}\label{eq:Galileo2}
  \frac{\partial}{\partial \tilde{x}{}^{\mu}} =
  b _{\mu}{}^{\rho} \frac{\partial}{\partial x^{\rho}}
  \: , \qquad
  \frac{\partial}{\partial \tilde{x}{}^0} = 
  \, \frac{1}{c} \, \big( \,  
  \frac{\partial}{\partial x^0} -
  v^{\rho} \frac{\partial}{\partial x^{\rho}} 
  \, \big) \: .
\end{equation}
Here $\mathbf{a}=(a^{\mu}{}_{\rho})$ is an invertible $3 \times 3$ matrix,
$\mathbf{b}=(b_{\mu}{}^{\rho})$ is the transpose of its inverse (i.e.
$\mathbf{a} \, \mathbf{b} ^T = \mathbf{b} ^T \mathbf{a} = \mathbf{1}$ 
or, in index notation, $a^{\mu}{}_{\rho} \, b _{\nu}{}^{\rho} = 
b_{\rho}{}^{\mu} a^{\rho}{}_{\nu} = \delta _{\nu}^{\mu}$),
$(v^1,v^2,v^3)$ is a real 3-tuple 
and $c$ is a non-zero real number. We call (\ref{eq:Galileo1})
and (\ref{eq:Galileo2}) a \emph{generalised Galilean transformation}.
It reduces to a standard Galilean transformation if $(a^{\mu}{}_{\rho})$
is orthogonal and $|c| = 1$.  An arbitrary coordinate transformation
induces on the cotangent and tangent spaces linear transformations
that can be written as a generalised Galilean transformation 
(\ref{eq:Galileo1}) and (\ref{eq:Galileo2}) followed by a 
transformation of the form
\begin{equation}\label{eq:change1}
  d \hat{x}{}^{\mu} = d \tilde{x}{}^{\mu}
  \: , \qquad
  d \hat{x}{}^0 = 
  d \tilde{x}{}^0 + u_{\sigma} d \tilde{x}{}^{\sigma} \, ,  
\end{equation}
\begin{equation}\label{eq:change2}
  \frac{\partial}{\partial \hat{x}{}^{\mu}} = 
  \frac{\partial}{\partial \tilde{x}{}^{\mu}} -
  u_{\mu} \, \frac{\partial}{\partial \tilde{x}{}^0}
  \: , \qquad
  \frac{\partial}{\partial \hat{x}{}^0} = 
  \frac{\partial}{\partial \tilde{x}{}^0}  
  \: ,
\end{equation}
where $(u_1, u_2 , u_3)$ is an arbitrary real 3-tuple. 

From the transformation behaviour (\ref{eq:even}) and (\ref{eq:odd})
of field strength and excitation we can calculate the transformation
behaviour of the fields $\vec{E}$, $\vec{B}$, $\vec{\mathcal{H}}$ and $\vec{D}$
as defined in (\ref{eq:EBHD}), and thereupon of the constitutive
matrix whose components are defined by 
(\ref{eq:impermittivity}) and (\ref{eq:impermeability}). With respect 
to generalised Galilean transformations (\ref{eq:Galileo1}) and 
(\ref{eq:Galileo2}), we find 
\begin{equation}\label{eq:GalEHDB}
\begin{split}
  \tilde{E}_{\nu} = \, \frac{1}{c} \, b_{\nu}{}^{\tau} 
  \big( E_{\tau} + v^{\rho} \epsilon _{\rho \tau \sigma} B^{\sigma} \big) 
  \: , \qquad
  \tilde{B}{}^{\mu} = \big( \mathrm{det} ( \mathbf{a} ) \big) ^{-1}
  a^{\mu}{}_{\nu} B^{\nu} \, , \qquad
\\
  \tilde{\mathcal{H}}_{\nu} = \, 
  \frac{\mathrm{det} ( \mathbf{a} )}{|c| \, | \mathrm{det} ( \mathbf{a} )|} \,
  b_{\nu}{}^{\tau} 
  \big( {\mathcal{H}}_{\tau} - v^{\rho} \epsilon _{\rho \tau \sigma} D^{\sigma} \big)
  \: , \qquad
  \tilde{D}{}^{\mu} = \, \frac{c}{|c| \, | \mathrm{det} ( \mathbf{a} ) |} \, 
  a^{\mu}{}_{\nu} D^{\nu} \, ,
\end{split}
\end{equation}
which yields, after some elementary algebra, the following transformation 
rule for the constitutive matrix.
\begin{equation}\label{eq:GalM}
  \begin{pmatrix} 
    \: \boldsymbol{\tilde{\kappa}} \: 
   & 
    \: \boldsymbol{\tilde{\chi}} \:
   \\[0.1cm]
    \: \boldsymbol{\tilde{\gamma}} \:
   & 
    \: \boldsymbol{\tilde{\nu}} \:
  \end{pmatrix}
  \; = \; 
  \frac{\mathrm{det}( \mathbf{a} )}{c} \; 
  \begin{pmatrix} 
    \frac{c \, \mathrm{det}( \mathbf{a} )}{| c \, \mathrm{det}( \mathbf{a} ) |}
    \, \mathbf{b} \, \boldsymbol{\kappa} \mathbf{b} ^T 
   & 
    \quad
    \, \mathbf{b} \, 
    \big( \boldsymbol{\chi} + v^{\sigma} \mathbf{A} _{\sigma} \big) 
    \, 
    \mathbf{b} ^T 
    \quad 
   \\[0.3cm]
    \quad
    \, \mathbf{b} \, 
    \big( \boldsymbol{\gamma} - v^{\sigma} \mathbf{A} _{\sigma} \big) 
    \, 
    \mathbf{b} ^T 
    \quad
   & 
    \frac{c \, \mathrm{det}( \mathbf{a} )}{| c \, \mathrm{det}( \mathbf{a} ) |}
    \, \mathbf{b} \, \boldsymbol{\nu} \mathbf{b} ^T 
  \end{pmatrix}
   \: .
\end{equation}
Here the antisymmetric matrices
$\mathbf{A}{}_{\rho} = \mathbf{A}{}^{\rho}$ are defined by (\ref{eq:Arho}).
By (\ref{eq:GalM}), $\mathbf{\tilde{M}}$ is well-defined
whenever $\mathbf{M}$ is. Hence, a generalised Galilean transformation 
transforms admissible coordinate systems into admissible coordinate systems.

With these results at hand, and with the transformation of the canonical
momentum coordinates
\begin{equation}\label{eq:ptrafo}
  {\tilde{p}}{}_{\mu} \, = \,
  b _{\mu}{}^{\rho} p_{\rho}  \: , \qquad
  \tilde{p}{}_0 \, = \, \frac{1}{c} \,
  \big( p_0 - v^{\rho} p_{\rho} \big) \; , 
\end{equation}
we can now calculate the transformation 
behaviour of the eigenvalues and eigenvectors of $p_{\rho} \mathbf{L} ^{\rho}$.
\begin{proposition}\label{prop:Galeigen}
  Under a generalised Galilean transformation (\ref{eq:Galileo1}) and (\ref{eq:Galileo2}),
  the eigenvalues and eigenvectors (\ref{eq:eigen}) of $p_{\rho} \mathbf{L}^{\rho}$
  transform according to
\begin{equation}\label{eq:Galeigen1}
   {\tilde{\omega}}{}_A (\vec{\tilde{p}}) 
   \, = \, 
   \frac{1}{c} \, 
   \big( \,
   \omega _A (\vec{p}) - v^{\rho} p_{\rho} 
   \, \big) \; ,
\end{equation}
\begin{equation}\label{eq:Galeigen2}
  \begin{pmatrix}
    \vec{\tilde{d}}{}_A ( \vec{\tilde{p}} ) 
  \\[0.2cm]
    \vec{\tilde{b}}_A ( \vec{\tilde{p}} ) 
  \end{pmatrix}
  \, = \, \mathrm{det} ( \mathbf{a} )^{-1} \, 
  \begin{pmatrix} 
    \mathbf{a} &  \mathbf{0} 
  \\[0.3cm]
    \mathbf{0} &  \mathbf{a}  
  \end{pmatrix}
  \, 
  \begin{pmatrix} 
    \frac{c \, \mathrm{det}( \mathbf{a} )}{| c \, \mathrm{det}( \mathbf{a} ) |}
    \, \vec{\, d}_A ( \vec{p} ) 
  \\[0.2cm] 
     \vec{\, b}_A ( \vec{p} ) 
  \end{pmatrix} 
\end{equation}
  for $A=1,2,3,4$.
\end{proposition}
\begin{proof}
The transformation behaviour (\ref{eq:GalM}) of the constitutive 
matrix, together with (\ref{eq:ptrafo}), allows us to calculate 
the transformed matrix  $\tilde{p}{}_{\rho} \tilde{\mathbf{L}}{}^{\rho}$, 
\begin{equation}\label{eq:GalpL}
  {\begin{pmatrix} 
    \mathbf{a} &  \mathbf{0} \\
    \mathbf{0} & \mathbf{a}
  \end{pmatrix}}^{-1}
  \,
  \tilde{p}{}_{\rho} \, {\tilde{\mathbf{L}}}{}^{\rho}
  \, 
  \begin{pmatrix} 
    \mathbf{a} &  \mathbf{0} \\
    \mathbf{0} & \mathbf{a}
  \end{pmatrix}
  \, = \, 
  \frac{1}{c} \, p_{\rho} \, 
  \begin{pmatrix} 
    \mathbf{0} &  - \mathbf{A}^{\rho} \\
    \mathbf{A}^{\rho} & \mathbf{0} 
  \end{pmatrix}
   \;
  \begin{pmatrix} 
    \frac{c \, \mathrm{det}( \mathbf{a} )}{| c \, \mathrm{det}( \mathbf{a} ) |}
    \, \boldsymbol{\kappa} 
   & 
    \quad
    \boldsymbol{\chi} + v^{\sigma} \mathbf{A} _{\sigma} 
    \quad 
   \\[0.3cm]
    \quad
    \boldsymbol{\gamma} - v^{\sigma} \mathbf{A} _{\sigma}  
    \, 
    \quad
   & 
    \frac{c \, \mathrm{det}( \mathbf{a} )}{| c \, \mathrm{det}( \mathbf{a} ) |}
    \, \boldsymbol{\nu} 
  \end{pmatrix}
  \; .
\end{equation}
Now assume that (\ref{eq:eigen}) is true. This eigenvalue equation can be 
equivalently rewritten as
\begin{equation}\label{eq:eigentrans}
  p_{\rho} \,
  \begin{pmatrix} 
    \mathbf{0} &  - \mathbf{A}^{\rho} \\
    \mathbf{A}^{\rho} & \mathbf{0} 
  \end{pmatrix}
  \begin{pmatrix} 
    \frac{c \, \mathrm{det}( \mathbf{a} )}{| c \, \mathrm{det}( \mathbf{a} ) |}
    \, \boldsymbol{\kappa} 
   & 
    \quad
    \boldsymbol{\chi}  
    \quad 
   \\[0.3cm]
    \quad
    \boldsymbol{\gamma}  
    \, 
    \quad
   & 
    \frac{c \, \mathrm{det}( \mathbf{a} )}{| c \, \mathrm{det}( \mathbf{a} ) |}
    \, \boldsymbol{\nu} 
  \end{pmatrix}
  \,
  \begin{pmatrix} 
    \frac{c \, \mathrm{det}( \mathbf{a} )}{| c \, \mathrm{det}( \mathbf{a} ) |}
    \, \vec{\, d}_A ( \vec{p} ) 
  \\[0.2cm] 
    \vec{\, b}_A ( \vec{p} ) 
  \end{pmatrix} 
   \; = \;
  \omega _A ( \vec{p})  \, 
  \begin{pmatrix} 
    \frac{c \, \mathrm{det}( \mathbf{a} )}{| c \, \mathrm{det}( \mathbf{a} ) |}
    \, \vec{\, d}_A ( \vec{p} ) 
  \\[0.2cm] 
    \vec{\, b}_A ( \vec{p} ) 
  \end{pmatrix} \; . 
\end{equation}
With ${\tilde{\omega}}{}_A (\vec{\tilde{p}})$, $\vec{\tilde{d}}{}_A ( \vec{\tilde{p}} )$ 
and $\vec{\tilde{b}}_A ( \vec{\tilde{p}} )$ introduced by (\ref{eq:Galeigen1}) and 
(\ref{eq:Galeigen2}), we find from (\ref{eq:GalpL}) and (\ref{eq:eigentrans}), with 
the help of Propositon \ref{prop:transverse}, that (\ref{eq:eigen}) holds with all 
terms twiddled.
\end{proof}

From this proposition we read that 
\begin{equation}\label{eq:invcone}
\tilde{p}{}_{\rho} \, d \tilde{x}{}^{\rho}
\, + \, {\tilde{\omega}}{}_A (\vec{\tilde{p}}) \, d \tilde{x}{}^0
\, = \, 
p_{\rho} \, dx^{\rho} \, + \, \omega _A (\vec{p}) \, dx^0
\end{equation}
which demonstrates that the light cones are invariant. The reduced characteristic
polynomial transforms as
\begin{equation}\label{eq:Galpol}
 \prod _{A=1} ^4 \big( \tilde{p}{}_0 - \tilde{\omega}{} _A (\vec{\tilde{p}} ) \big)
  \; = \; \frac{1}{c^4} \;
 \prod _{A=1} ^4 \big( p_0 - \omega _A (\vec{p} ) \big) \, .
\end{equation}
With
\begin{equation}\label{eq:Galkappa}
\mathrm{det} ( \boldsymbol{\tilde{\kappa}} )
\; = \;
\frac{| \mathrm{det}( \mathbf{a} ) |}{|c|} \, 
\mathrm{det}( \boldsymbol{\kappa}) \; ,
\end{equation}
which follows from (\ref{eq:GalM}), this implies
\begin{equation}\label{eq:Galpol2}
 \frac{1}{\mathrm{det}(\boldsymbol{\tilde{\kappa}})} \,
 \prod _{A=1} ^4 \big( \tilde{p}{}_0 - \tilde{\omega}{} _A (\vec{\tilde{p}} ) \big)
  \; = \; \frac{1}{|c \, \mathrm{det}( \mathbf{a})|} \; 
  \frac{1}{\mathrm{det} (\boldsymbol{\kappa})} \;
  \prod _{A=1} ^4 \big( p_0 - \omega _A (\vec{p} ) \big) \, .
\end{equation}
This shows that under generalised Galilean transformations the $\mathcal{G}^{abcd}$,
introduced in (\ref{eq:TR}), transform as a tensor density of weight 1,
\begin{equation}\label{eq:GalTR}
\tilde{\mathcal{G}}{}^{abcd} \, \tilde{p}{}_a \, \tilde{p}{}_b
\, \tilde{p}{}_c\, \tilde{p}{}_d
\; = \;
\Big| \, \mathrm{det} \Big( \frac{\partial x}{\partial \tilde{x}} \Big) \, \Big| \;
\mathcal{G}^{abcd} \, p_a \, p_b \, p_c \, p_d \; .
\end{equation}

We now turn to transformations of the form (\ref{eq:change1}) and 
(\ref{eq:change2}). The fields change according to
\begin{equation}\label{eq:changeEHDB}
\begin{split}
  \hat{E}_{\nu} = \tilde{E}_{\nu}
  \: , \qquad
  \hat{B}{}^{\mu} = 
  \tilde{B}{}^{\mu} + 
  \epsilon^{\mu \nu \sigma} u_{\nu} \tilde{E}_{\sigma} \, , \:
\\
  \hat{\mathcal{H}}_{\nu} = \tilde{\mathcal{H}}_{\nu} 
  \: , \qquad
  \hat{D}{}^{\mu} = 
  \tilde{D}{}^{\mu} - 
  \epsilon^{\mu \nu \sigma} u_{\nu} \tilde{\mathcal{H}}_{\sigma} \, ,
\end{split}
\end{equation}
where $\epsilon ^{\mu \nu \sigma}$ is the contravariant Levi-Civita
symbol, defined by the properties that it is totally antisymmetric
and satisfies $\epsilon ^{123} = 1$.
From this we find the following transformation behaviour of
the constitutive matrix.
\begin{equation}\label{eq:changeM}
  \begin{pmatrix} 
    \: \boldsymbol{\hat{\kappa}} \: 
   & 
    \: \boldsymbol{\hat{\chi}} \:
   \\[0.1cm]
    \: \boldsymbol{\hat{\gamma}} \:
   & 
    \: \boldsymbol{\hat{\nu}} \:
  \end{pmatrix}
  = \, 
  \begin{pmatrix} 
    \: \boldsymbol{\tilde{\kappa}} \: 
   & 
    \: \boldsymbol{\tilde{\chi}} \:
   \\[0.1cm]
    \: \boldsymbol{\tilde{\gamma}} \:
   & 
    \: \boldsymbol{\tilde{\nu}} \:
  \end{pmatrix}
    \; \left( \;
  \begin{pmatrix} 
    \: \mathbf{1} \: 
   & 
    \: \mathbf{0} \:
   \\[0.1cm]
    \: \mathbf{0} \:
   & 
    \: \mathbf{1} \:
  \end{pmatrix}
   \, - \, u_{\rho}
  \, 
  \begin{pmatrix} 
    \: \mathbf{0} \: 
   & 
    \: - \mathbf{A} ^{\rho} \:
   \\[0.1cm]
    \: \mathbf{A} ^{\rho} \:
   & 
    \: \mathbf{0} \:
  \end{pmatrix}
   \;
  \begin{pmatrix} 
    \: \boldsymbol{\tilde{\kappa}} \: 
   & 
    \: \boldsymbol{\tilde{\chi}} \:
   \\[0.1cm]
    \: \boldsymbol{\tilde{\gamma}} \:
   & 
    \: \boldsymbol{\tilde{\nu}} \:
  \end{pmatrix}
  \, \right) ^{-1} \: .
\end{equation}
Thus, a coordinate transformation of the form (\ref{eq:change1}) and 
(\ref{eq:change2}) maps admissible coordinates into admissible 
coordinates if and only if the inverse matrix on the right-hand side
of (\ref{eq:changeM}) exists. This is true for almost all values of
$(u_1, u_2, u_3)$, namely whenever $d \tilde{x}{} ^0+u_{\rho}
d \tilde{x}{}^{\rho}$ is
non-characteristic.
 
From (\ref{eq:changeM}) and the transformation
behaviour of the momentum coordinates,
\begin{equation}\label{eq:changep}
  \hat{p}{}_{\mu} = 
  \tilde{p}{}_{\mu} -
  u_{\mu} \,  \tilde{p}{}_0
  \: , \qquad
  \hat{p}{}_0 = 
  \tilde{p}{}_0  
  \: ,
\end{equation}
we find the transformation behaviour of the 
characteristic matrix,
\begin{equation}\label{eq:changechar}
  \hat{p}{}_0 \mathbf{1} \, - \, 
  \hat{p}{}_{\rho} \mathbf{\hat{L}}{}^{\rho}
  \; = \;
  \big( \, 
  \tilde{p}{}_0 \mathbf{1} \, - \, 
  \tilde{p}{}_{\rho} \mathbf{\tilde{L}}{}^{\rho}
  \big) \;
  \big(
  \mathbf{1} \, - \, 
  u_{\sigma} \mathbf{\tilde{L}}{}^{\sigma}
  \big) ^{-1}
  \; .
\end{equation}
This demonstrates that the characteristic equations
$\mathrm{det}(  \hat{p}{}_0 \mathbf{1} \, - \, 
\hat{p}{}_{\rho} \mathbf{\hat{L}}{}^{\rho} ) =0$ 
and $\mathrm{det}( \tilde{p}{}_0 \mathbf{1} \, - \, 
\tilde{p}{}_{\rho} \mathbf{\tilde{L}}{}^{\rho} )=0$ 
are equivalent, i.e., that the light cones
are invariant with respect to coordinate transformations
of the form (\ref{eq:change1}) and (\ref{eq:change2}).
To verify the transformation behaviour of the Tamm-Rubilar 
tensor density, we assume that $\boldsymbol{\tilde{\kappa}}$
and $\mathbf{\tilde{M}}$ are invertible. Then (\ref{eq:changeM})
takes the form
\begin{equation}\label{eq:changeinv}
  \mathbf{\hat{M}}{}^{-1} \; = \; \mathbf{\tilde{M}}{}^{-1} \; - \;
  u_{\rho} \,
  \begin{pmatrix} 
    \: \mathbf{0} \: 
   & 
    \: - \mathbf{A} ^{\rho} \:
   \\[0.1cm]
    \: \mathbf{A} ^{\rho} \:
   & 
    \: \mathbf{0} \: 
  \end{pmatrix}
  \: .
\end{equation}
After calculating the inverse matrices $\mathbf{\tilde{M}}{}^{-1}$ and 
$\mathbf{\hat{M}}{}^{-1}$ with the help of (\ref{eq:transM}), the lower 
right-hand block of (\ref{eq:changeinv}) yields
\begin{equation}\label{eq:comp}
\boldsymbol{\hat{\nu}} - \boldsymbol{\hat{\gamma}}
\boldsymbol{\hat{\kappa}}{}^{-1} \boldsymbol{\hat{\chi}}
\; = \; 
\boldsymbol{\tilde{\nu}} - \boldsymbol{\tilde{\gamma}}
\boldsymbol{\tilde{\kappa}}{}^{-1} \boldsymbol{\tilde{\chi}} \; ,
\end{equation}
which, by (\ref{eq:detM}), implies
\begin{equation}\label{eq:compdet}
\mathrm{det}(\mathbf{\hat{M}}) \; \mathrm{det}(\tilde{\kappa})
\; = \; 
\mathrm{det}(\mathbf{\tilde{M}}) \; \mathrm{det}(\hat{\kappa})
\; .
\end{equation}
On the other hand, comparison of (\ref{eq:changeM}) and  
(\ref{eq:changechar}) yields
\begin{equation}\label{eq:transdet}
\mathrm{det}(  \hat{p}{}_0 \mathbf{1} \, - \, 
  \hat{p}{}_{\rho} \mathbf{\hat{L}}{}^{\rho})
  \; 
  \mathrm{det}(\mathbf{\tilde{M}}) 
  \; = \; 
  \mathrm{det}(  \tilde{p}{}_0 \mathbf{1} \, - \, 
  \tilde{p}{}_{\rho} \mathbf{\tilde{L}}{}^{\rho} )
  \; \mathrm{det}(\mathbf{\hat{M}}) 
\end{equation}
and thus, after dividing by $\hat{p}{}_0^2 = \tilde{p}{}_0^2$ and
using (\ref{eq:compdet}),
\begin{equation}\label{eq:changepol}
 \prod _{A=1} ^4 \big( \hat{p}{}_0 - \hat{\omega}{} _A (\vec{\hat{p}} ) \big)
  \; \mathrm{det}(\mathbf{\tilde{\kappa}}) 
 \; = \; 
 \prod _{A=1} ^4 \big( \tilde{p}{}_0 - \tilde{\omega}{} _A (\vec{\tilde{p}} ) \big)
  \; \mathrm{det}(\mathbf{\hat{\kappa}}) \; . 
\end{equation}
This shows that, also with respect to transformations (\ref{eq:change1}) and
(\ref{eq:change2}), the $\mathcal{G}^{abcd}$ of (\ref{eq:TR}) transform as 
a tensor density of weight 1,
\begin{equation}\label{eq:changeTR}
\hat{\mathcal{G}}{}^{abcd} \, \hat{p}{}_a \, \hat{p}{}_b
\, \hat{p}{}_c\, \hat{p}{}_d
\; = \;
\Big| \, \mathrm{det} \Big( \frac{\partial \tilde{x}}{\partial \hat{x}} 
\Big) \, \Big| \;
\tilde{\mathcal{G}}{}^{abcd} \, \tilde{p}{}_a \, \tilde{p}{}_b \, 
\tilde{p}{}_c \, \tilde{p}{}_d \; ,
\end{equation}
where in this case the determinant on the right-hand side is equal to 1.

\section{Calculating the roots of the characteristic equation}\label{sec:roots}
\noindent
In this section we will discuss the question of how to calculate the roots
of the characteristic equation if the constitutive matrix is given.

$\omega _1 (\vec{p})$, $\omega _2 (\vec{p})$, $\omega _3 (\vec{p})$ and 
$\omega _4 (\vec{p})$ are, together with $\omega _5 (\vec{p})=\omega _6 (\vec{p})=0$, 
the six eigenvalues of the matrix $p_{\rho} \mathbf{L}^{\rho}$. In general
they are complex. Now the sum of all eigenvalues is the trace, the
sum of the squares of the eigenvalues is the trace of the square, and so on. 
(This is obvious in the case of a diagonizable matrix. It is also true in 
general, as can be seen from the Jordan decomposition theorem.) Thus,
if we define
\begin{equation}\label{eq:traceL}
  L^{\rho_1 \cdots \rho _i} = \mathrm{trace} \big( 
  \mathbf{L} ^{(\rho _1} \cdots \mathbf{L} ^{\rho _i)} \big) \: ,
\end{equation}
where round brackets around indices mean symmetrization, we get 
\begin{equation}\label{eq:4omega}
\begin{split}
  p_{\rho} L^{\rho}  & =   
  \omega _1 (\vec{p}) + \omega _2 (\vec{p}) + \omega _3 (\vec{p}) + \omega _4 (\vec{p})
   \: , \\ 
  p_{\rho} p _ {\sigma} L^{\rho \sigma}  & =  
  \omega _1 (\vec{p})^2 + \omega _2 (\vec{p})^2 + \omega _3 (\vec{p})^2 + \omega _4 (\vec{p})^2
  \: , \\ 
  p_{\rho} p _ {\sigma} p_ {\tau} L^{\rho \sigma \tau}  & =  
  \omega _1 (\vec{p})^3 + \omega _2 (\vec{p})^3 + \omega _3 (\vec{p})^3 + \omega _4 (\vec{p})^3
  \: , \\ 
  p_{\rho} p _ {\sigma} p _ {\tau} p_{\lambda} L^{\rho \sigma \tau \lambda}  & =  
  \omega _1 (\vec{p})^4 + \omega _2 (\vec{p})^4 + \omega _3 (\vec{p})^4 + \omega _4 (\vec{p})^4
  \: . 
\end{split}
\end{equation}
If the constitutive matrix is known, $\mathbf{L}^{\rho}$ can be calculated from (\ref{eq:Lrho})
and the $L^{\rho_1 \cdots \rho _i}$ can be calculated from (\ref{eq:traceL}). Then, the four
equations (\ref{eq:4omega}) of (maximal) order four determine the four roots $\omega _A
(\vec{p})$. In general, solving fourth-order equations leads to rather awkward expressions.
For several special cases, however, this method allows to calculate the roots of the characteristic
equation in a convenient way, as will be demonstrated below.

The transformation behaviour (\ref{eq:Galeigen1}) of the $\omega _A (\vec{p} )$ implies 
the following transformation behaviour of the $L^{\rho_1 \cdots \rho _i}$ under
generalised Galilean transformations (\ref{eq:Galileo1}) and (\ref{eq:Galileo2}).
\begin{equation}\label{eq:4GalL}
\begin{split}
  c \, \tilde{p}{}_{\rho} \tilde{L}{}^{\rho}  & =   
  p_{\rho} \big(  L^{\rho} + 4 v^{\rho} \big)
   \: , \\ 
  c^2 \, \tilde{p}{}_{\rho} \tilde{p}{}_ {\sigma} \tilde{L}{}^{\rho \sigma}  & =  
  p_{\rho} p_{\sigma} \big( L^{\rho \sigma} + 2 L^{\rho} v^{\sigma} 
  + 4 v^{\rho} v^{\sigma} \big) 
  \: , \\ 
  c^3 \, \tilde{p}{}_{\rho} \tilde{p}{}_ {\sigma} \tilde{p}{}_ {\tau} 
  \tilde{L}{}^{\rho \sigma \tau}  & =  
  p_{\rho} p_{\sigma} p_{\tau} \big( L^{\rho \sigma \tau} + 
  3 L^{\rho \sigma} v^{\tau} + 3 L^{\rho} v^{\sigma} v^{\tau} + 
  4 v^{\rho} v^{\sigma} v^{\tau} \big)
  \: , \\ 
  c^4 \, \tilde{p}{}_{\rho} \tilde{p}{}_ {\sigma} \tilde{p}{}_ {\tau} \tilde{p}{}_{\lambda} 
  \tilde{L}{}^{\rho \sigma \tau \lambda}  & =  
  p_{\rho} p_{\sigma} p_{\tau} p_{\lambda} \big( L^{\rho \sigma \tau \lambda} 
  + 4 L^{\rho \sigma \tau} v^{\lambda} + 6 L^{\rho \sigma} v^{\tau} v^{\lambda} 
  + 4 L^{\rho} v ^{\sigma} v^{\tau} v^{\lambda} + 
  4 v^{\rho} v^{\sigma} v^{\tau} v^{\lambda} \big)
  \: . 
\end{split}
\end{equation}
As an alternative, one can derive (\ref{eq:4GalL}) by multiplying each side of
(\ref{eq:GalpL}) sufficiently often with itself and then calculating the 
trace. However, this is much more tedious than using (\ref{eq:Galeigen1}).
 
From the first equation of (\ref{eq:4GalL}) we read that, by a generalised Galilean transformation
with $v^{\rho} = -\frac{1}{4} L^{\rho}$, it is always possible to transform $L^{\rho}$ to 
$\tilde{L}{}^{\rho} = 0$. Also, we read from (\ref{eq:4GalL}) that with respect to 
purely spatial transformations ($v^{\rho} = 0$, $c=1$) the $L^{\rho_1 \cdots \rho _i}$ 
behave as contravariant tensor components.

The characteristic equation is uniquely determined by the coefficients $L^{\rho}, L^{\rho \sigma},
L^{\rho \sigma \tau}, L^{\rho \sigma \tau \lambda}$. As they are totally symmetric, these are
$(3+6+10+15)=34$ independent real numbers. On the other hand, the constitutive matrix has $(6 \times 6)
= 36$ independent components. From this observation it follows that different constitutive matrices
must yield the same characteristic equation. It is an interesting and important problem to
find a necessary and sufficient condition for two constitutive matrices to give the same characteristic
equation. This problem is unsolved so far; however, in the next section we will find a partial
answer by determining a group action on the set of all constitutive matrices that leaves the 
characteristic equation invariant.

\section{$SL(2 ,\mathbb{R})$ action on constitutive matrices}\label{sec:SL2R}
\noindent
We consider the group $SL(2, \mathbb{R})$ in terms of its natural representation
by $6 \times 6$ matrices,
\begin{equation}\label{eq:SL2R}
SL ( 2 , \mathbb{R} ) \, = \, \Big\{ \, 
\begin{pmatrix}
\, a \, \mathbf{1} \, & \, b \, \mathbf{1} \,
\\
\, c \, \mathbf{1} \, & \, d \, \mathbf{1} \,  
\end{pmatrix}
\: \Big| \: ad-bc = 1 \: \Big\} \; .
\end{equation}
This group acts on the set of all real $6 \times 6$ matrices
by conjugation, i.e., each element $\mathbf{Q} \in SL(2 , \mathbb{R} )$
maps each constitutive matrix $\mathbf{M}$ onto $\mathbf{M'} =
\mathbf{Q} ^T \mathbf{M} \mathbf{Q}$. In terms of $ 3 \times 3$ blocks,
the group action reads
\begin{equation}\label{eq:MM'}
\begin{pmatrix}
\, \boldsymbol{\kappa} \, & \, \boldsymbol{\chi} \,
\\
\, \boldsymbol{\gamma} \, & \, \boldsymbol{\nu} \,  
\end{pmatrix}
\: \longmapsto \:
\begin{pmatrix}
\, \boldsymbol{\kappa '} \, & \, \boldsymbol{\chi '} \,
\\
\, \boldsymbol{\gamma '} \, & \, \boldsymbol{\nu '} \,  
\end{pmatrix}
\; = \; 
\begin{pmatrix}
\, a^2 \boldsymbol{\kappa} + ac (\boldsymbol{\chi} + \boldsymbol{\gamma} )
+ c^2 \boldsymbol{\nu} \: & 
\: ab \boldsymbol{\kappa} + ad \boldsymbol{\chi} + bc \boldsymbol{\gamma} 
+ cd \boldsymbol{\nu} \,
\\[0.2cm]
\, ab \boldsymbol{\kappa} + bc \boldsymbol{\chi} + ad \boldsymbol{\gamma} 
+ cd \boldsymbol{\nu} \: & 
\: b^2 \boldsymbol{\kappa} + bd (\boldsymbol{\chi} + \boldsymbol{\gamma} )
+ d^2 \boldsymbol{\nu} \,
\end{pmatrix} \: .
\end{equation}
It is obvious that this group action leaves the determinant invariant,
$\mathrm{det} ( \mathbf{M'} )=\mathrm{det} ( \mathbf{M} ) $, that it maps symmetric 
matrices onto symmetric matrices and that it preserves the difference of the 
off-diagonal blocks, $\boldsymbol{\chi '} - \boldsymbol{\gamma '}=
\boldsymbol{\chi} - \boldsymbol{\gamma}$.
The following calculation shows that the group action leaves the 
characteristic equation invariant.
\begin{gather}
\nonumber
\mathrm{det} \Big( \; p_0 \, \mathbf{1} \, - \, 
\begin{pmatrix}
\mathbf{0} & \; - \, p_{\rho} \mathbf{A} ^{\rho} \;
\\
\: p_{\rho} \mathbf{A} ^{\rho} \: & \mathbf{0}
\end{pmatrix}
\: \mathbf{Q} ^T \; \mathbf{M} \; \mathbf{Q} \: \Big)
\; = 
\\
\label{eq:charinv}
\mathrm{det} \Big( \; p_0 \, \mathbf{1} \, - \, 
\mathbf{Q} \: 
\begin{pmatrix}
\mathbf{0} & \; - \, p_{\rho} \mathbf{A} ^{\rho} \;
\\
\: p_{\rho} \mathbf{A} ^{\rho} \: & \mathbf{0}
\end{pmatrix}
\: \mathbf{Q} ^T \; \mathbf{M} \; \Big)
\; = 
\\
\nonumber
\mathrm{det} \Big( \; p_0 \, \mathbf{1} \, - \, 
\begin{pmatrix}
\mathbf{0} & \; - \, p_{\rho} \mathbf{A} ^{\rho} \;
\\
\: p_{\rho} \mathbf{A} ^{\rho} \: & \mathbf{0}
\end{pmatrix}
\; \mathbf{M} \; \Big) \: .
\end{gather}
In the first step we have used the Sylvester identity according
to which $\mathrm{det} ( s \mathbf{1} - \mathbf{A} \mathbf{B} )
= \mathrm{det} ( s \mathbf{1} - \mathbf{B} \mathbf{A} )$ for all
scalars $s$ and all $n \times n$ matrices $\mathbf{A}$ and
$\mathbf{B}$. In the second step we have used that
\begin{equation}\label{eq:invA}
\begin{pmatrix}
\, a \, \mathbf{1} \, & \, b \, \mathbf{1} \,
\\
\, c \, \mathbf{1} \, & \, a \, \mathbf{1} \,  
\end{pmatrix}
\:
\begin{pmatrix}
\mathbf{0} & \; - \, p_{\rho} \mathbf{A} ^{\rho} \;
\\
\: p_{\rho} \mathbf{A} ^{\rho} \: & \mathbf{0}
\end{pmatrix}
\:
\begin{pmatrix}
\, a \, \mathbf{1} \, & \, c \, \mathbf{1} \,
\\
\, b \, \mathbf{1} \, & \, a \, \mathbf{1} \,  
\end{pmatrix}
\; = \;
\begin{pmatrix}
\mathbf{0} & \; - \, p_{\rho} \mathbf{A} ^{\rho} \;
\\
\: p_{\rho} \mathbf{A} ^{\rho} \: & \mathbf{0}
\end{pmatrix}
\end{equation}
if $ad-bc =1$, as can be quickly verified by multiplying out the 
the left-hand side.

As $SL(2 , \mathbb{R})$ is 3-dimensional, the orbits of the group
action must be of dimension $\le 3$. To calculate the dimension of
the orbits we have to differentiate the group action. A quick 
calculation shows that the tangent space to the orbit through
the matrix $\mathbf{M}$ with $3 \times 3$ blocks according to
(\ref{eq:defM}) is spanned by the three matrices
\begin{equation}\label{eq:Torbit}
\mathbf{E} _1 \, = \, 
\begin{pmatrix}
\; \boldsymbol{\kappa} \; & \mathbf{0}
\\
\mathbf{0} & \, \boldsymbol{\nu} \,
\end{pmatrix}
\; , \qquad
\mathbf{E} _2 \, = \, 
\begin{pmatrix}
\, \boldsymbol{\gamma} +  \boldsymbol{\chi} \, & \: \boldsymbol{\nu} \;
\\
\boldsymbol{\nu} & \,  \mathbf{0} \,
\end{pmatrix}
\; , \qquad
\mathbf{E} _1 \, = \, 
\begin{pmatrix}
\, \mathbf{0} \, & \boldsymbol{\kappa}
\\
\: \boldsymbol{\kappa} \: & \, \boldsymbol{\gamma} +  \boldsymbol{\chi}
\end{pmatrix}
\; .
\end{equation}
$\mathbf{E} _1 , \mathbf{E} _2$ and $\mathbf{E} _3$ are linearly independent unless 
one of the three matrices $\boldsymbol{\kappa}, \boldsymbol{\nu}$ and 
$\boldsymbol{\gamma} +  \boldsymbol{\chi}$ is zero and the other two are
linearly dependent. This demonstrates that the group action foliates
a dense and open subset of the set of all real $6 \times 6$ matrices
into three-dimensional orbits. 

The $SL(2 , \mathbb{R} )$ transformations on constitutive matrices contain two 
interesting special examples. The first is the one-parameter family of transformations
with $a = b = 0$ and $b = c^{-1}$ which corresponds to the so-called \emph{reciprocity 
transformations}. By definition, a reciprocity transformation is a transformation 
$F_{ab} \mapsto - \frac{1}{\zeta} H_{ab}$, $H_{ab} \mapsto
\zeta F_{ab}$ of field strength and excitation, where $\zeta$ 
is a nowhere vanishing pseudoscalar field. A reciprocity 
transformation changes 
\begin{equation}\label{eq:recfields}
  \vec{E} \mapsto \frac{1}{\zeta} \vec{\mathcal{H}} \: , \qquad
  \vec{B} \mapsto -\frac{1}{\zeta} \vec{D} \: , \qquad
  \vec{\mathcal{H}} \mapsto - \zeta \vec{E} \: , \qquad
  \vec{D} \mapsto \zeta \vec{B} \: ,
\end{equation}
and, thus, the constitutive matrix according to 
\begin{equation}\label{eq:reccon}
\begin{pmatrix}
\: \boldsymbol{\kappa} \: & \: \boldsymbol{\chi} \:
\\
\: \boldsymbol{\gamma} \: & \: \boldsymbol{\nu} \:
\end{pmatrix}
\: \longmapsto \:
\begin{pmatrix}
\: \frac{1}{\zeta ^2} \boldsymbol{\nu} \: & \: - \boldsymbol{\gamma} \:
\\
\: - \boldsymbol{\chi}  \: & 
\: \zeta ^2 \boldsymbol{\kappa} \:
\end{pmatrix}
\: .
\end{equation}
This is precisely the transformation produced by the $SL(2 , \mathbb{R})$
element with $a = b = 0$ and $b = c^{-1} = \zeta$. 
So our result contains as a special case the fact that the characteristic equation is invariant
under reciprocity transformations (cf. Hehl and Obukhov \cite{HehlObukhov2003}, pp 273). 
Note, however, that in the case of a nonlinear constitutive law a reciprocity transformation 
changes the argument of the constitutive matrix.

The second interesting special case is the one-parameter family of transformations
with $a=d=1$ and $c=0$. This corresponds to adding an \emph{axion} field, i.e.,
to a transformation $F_{ab} \mapsto F_{ab}$, $H_{ab} \mapsto H_{ab} + \phi F_{ab}$
with a pseudoscalar field $\phi$. 
This transformation changes
\begin{equation}\label{eq:addaxion}
\vec{E} \mapsto \vec{E} \, , \qquad
\vec{B} \mapsto \vec{B} \, , \qquad
\vec{D} \mapsto \vec{D} + \phi \vec{B}  \, , \qquad
\vec{\mathcal{H}} \mapsto \vec{\mathcal{H}} - \phi \vec{E}  \, ,
\end{equation}
and, thus, the constitutive matrix according to 
\begin{equation}\label{eq:axionM}
\begin{pmatrix}
\: \boldsymbol{\kappa} \: & \: \boldsymbol{\chi} \:
\\
\: \boldsymbol{\gamma} \: & \: \boldsymbol{\nu} \:
\end{pmatrix}
\: \longmapsto \:
\begin{pmatrix}
\: \boldsymbol{\kappa} \: & \: \boldsymbol{\chi} + \phi \, \boldsymbol{\kappa} \:
\\
\: \boldsymbol{\gamma} + \phi \, \boldsymbol{\kappa}  \: & 
\: \boldsymbol{\nu} + \phi ( \boldsymbol{\gamma} + \boldsymbol{\chi}) + \phi ^2 \boldsymbol{\kappa} \:
\end{pmatrix}
\: .
\end{equation}
This is precisely the transformation produced by the $SL(2, \mathbb{R})$ 
element with $a=d=1$, $c=0$ and $b=\phi$. We have thus reproduced the 
known fact (see Hehl and Obukhov \cite{HehlObukhov2003}, p. 265)
that adding an axion field does not affect the characteristic equation.

\section{Reduction to 3 dimensions}\label{sec:reduction}
\noindent
On the left-hand side of the characteristic equation (\ref{eq:characteristic})
we have the determinant of a $6 \times 6$ matrix. If the impermittivity matrix 
$\boldsymbol{\kappa}$  or the impermeability matrix $\boldsymbol{\nu}$ is 
invertible, this may be reduced to the determinant of a $3 \times 3$ matrix.
We give the derivation for the case that $\boldsymbol{\kappa}$ is invertible. 
We may use the decomposition (\ref{eq:transM}) of the constitutive
matrix. If we feed this into the characteristic equation (\ref{eq:characteristic}),
and apply Sylvester's formula that $\mathrm{det} ( s \mathbf{1} - \mathbf{A}  
\mathbf{B} ) \, = \,  \mathrm{det} ( s \mathbf{1} - \mathbf{B}  \mathbf{A} )$
for all scalars $s$ and all $n \times n$ matrices $\mathbf{A}$ and $\mathbf{B}$,
we find
\begin{gather}\label{eq:red}
  0 = \mathrm{det} 
  \left( \: p_0 \; \mathbf{1} \: - \:
  \begin{pmatrix} 
    \: \mathbf{1} \: & \: \boldsymbol{\kappa} ^{-1} \boldsymbol{\chi} \:
    \\[0.1cm]
    \mathbf{0} & \mathbf{1}
  \end{pmatrix}
  \begin{pmatrix}
    \mathbf{0}   & \: - p_{\rho} \mathbf{A} ^{\rho} \: 
    \\[0.2cm]
    \: p_{\rho} \mathbf{A} ^{\rho} \: & \mathbf{0}
  \end{pmatrix} \: 
  \begin{pmatrix} 
    \mathbf{1} & \: \mathbf{0} \:
    \\[0.1cm]
    \: \boldsymbol{\gamma} \boldsymbol{\kappa} ^{-1} \:
    & \mathbf{1}
  \end{pmatrix}
  \begin{pmatrix} 
    \: \boldsymbol{\kappa} \: 
    &  \mathbf{0} 
    \\[0.1cm]
    \mathbf{0} & \: \boldsymbol{\nu} - 
    \boldsymbol{\gamma} \boldsymbol{\kappa} ^{-1} \boldsymbol{\chi} \: 
  \end{pmatrix}
  \:   \right)
\\[0.4cm]
\qquad = \:
\mathrm{det} 
  \begin{pmatrix} 
  \quad   p_0 \mathbf{1} - p_{\rho} ( \boldsymbol{\kappa} ^{-1} 
   \boldsymbol{\chi} \mathbf{A} ^{\rho} - \mathbf{A} ^{\rho} \boldsymbol{\gamma}
   \boldsymbol{\kappa} ^{-1} ) \boldsymbol{\kappa} \quad 
   &  
   \quad p_{\sigma} \mathbf{A} ^{\sigma} ( \boldsymbol{\nu} - \boldsymbol{\gamma}
   \boldsymbol{\kappa} ^{-1} \boldsymbol{\chi} ) \quad 
   \\[0.1cm]
   - p_{\rho} \mathbf{A} ^{\rho} \boldsymbol{\kappa}
   & 
   p_0 \mathbf{1}
  \end{pmatrix}
  \: .
\end{gather}
We now use the well-known rule (see e.g. \cite{Sylvester2000}) that for
any $n \times n$ matrices $\mathbf{A}$, $\mathbf{B}$, $\mathbf{C}$, $\mathbf{D}$
\begin{equation}\label{eq:block}
  \mathrm{det}
  \begin{pmatrix}
    \mathbf{A} & \mathbf{B} \\ \mathbf{C} & \mathbf{D}
  \end{pmatrix}
  \, = \, \mathrm{det}
  \big( \mathbf{A}\mathbf{D}-\mathbf{B}\mathbf{C} \big)
  \qquad \text{if} \qquad
   \mathbf{C} \, \mathbf{D} \, = \,  \mathbf{D} \, \mathbf{C} \; .
  \end{equation}
This puts the characteristic equation into the form
\begin{equation}\label{eq:char2}
  0 = \mathrm{det}
  \big( \, p_0 ^2 \, \boldsymbol{\kappa}^{-1} + 
  p_0p_{\rho} ( \mathbf{A}^{\rho} \boldsymbol{\gamma} \boldsymbol{\kappa} ^{-1}
  - \boldsymbol{\kappa}^{-1} \boldsymbol{\chi} \mathbf{A}^{\rho} ) +
  p_{\rho} p_{\sigma} \mathbf{A} ^{\rho} 
  ( \boldsymbol{\nu} - \boldsymbol{\gamma} \boldsymbol{\kappa}^{-1} \boldsymbol{\chi} )
  \mathbf{A}^{\sigma} \, \big) \: .
\end{equation}
If $\boldsymbol{\nu}$ is invertible, an analogous calculation results in
\begin{equation}\label{eq:char1}
  0 = \mathrm{det}
  \big( \, p_0 ^2 \, \boldsymbol{\nu}^{-1} - 
  p_0 p_{\rho} ( \mathbf{A}^{\rho} \boldsymbol{\chi} \boldsymbol{\nu} ^{-1}
  - \boldsymbol{\nu}^{-1} \boldsymbol{\gamma} \mathbf{A}^{\rho} ) +
  p_{\rho} p_{\sigma} \mathbf{A} ^{\rho} 
  ( \boldsymbol{\kappa} - \boldsymbol{\chi} \boldsymbol{\nu}^{-1} \boldsymbol{\gamma} )
  \mathbf{A}^{\sigma} \, \big) \: .
\end{equation}
Note that, by a reciprocity transformation (\ref{eq:reccon}), equation (\ref{eq:char2}) 
transforms into (\ref{eq:char1}) and vice versa. Thus, if both 
$\boldsymbol{\nu}$ and $\boldsymbol{\kappa}$ are invertible, (\ref{eq:char2}) and 
(\ref{eq:char1}) are indeed equivalent forms of the characteristic equation. 

If the constitutive matrix (\ref{eq:defM}) is invertible, (\ref{eq:char2}) and 
(\ref{eq:char1}) are equivalent to the form derived by Graglia, 
Uslenghi and Zich \cite{GragliaUslenghiZich1991}, eq. (7) and (8),
apart from the fact that they considered only linear constitutive
laws. If the cross-terms 
$\boldsymbol{\chi}$ and $\boldsymbol{\gamma}$ vanish, they reduce to the form of 
Damaskos, Maffett and Uslenghi \cite{DamaskosMaffettUslenghi1982}, eq. (7). An alternative version of the characteristic equation,
for linear constitutive laws without cross-terms, was derived and
discussed by Itin \cite{Itin2010}.

\section{Invariance under time and space inversion}\label{sec:time}
\noindent
In general, the characteristic equation is not invariant under time inversion
$(p_0, p_1, p_2, p_3) \mapsto (-p_0, p_1, p_2, p_3)$. Similarly, it is not
invariant under space inversion $(p_0, p_1, p_2, p_3) \mapsto (p_0, -p_1, -p_2, -p_3)$.
However, owing to the homogeneity of the characteristic polynomial, it is invariant under
combined time and space inversion, $(p_0, p_1, p_2, p_3) \mapsto (-p_0, -p_1, -p_2, -p_3)$.
This implies that the characteristic equation is invariant under time inversion if
and only if it is invariant under space inversion. It is easy to see from 
(\ref{eq:characteristic}), and even more obvious from (\ref{eq:char2}) or 
(\ref{eq:char1}), that a sufficient condition for invariance under time inversion
is that the magneto-electric cross-terms $\boldsymbol{\chi}$ and $\boldsymbol{\gamma}$ 
vanish. However, this is not necessary. 

Clearly, the characteristic equation is invariant under time inversion if and
only if its roots coincide pairwise up to sign, $\omega _4 = - \omega _1$
and $\omega _3 = -\omega _2$. From (\ref{eq:4omega}) we read that this is 
true if and only if $L^{\rho} = 0$ and $L^{\rho \sigma \tau}=0$. In this special 
case (\ref{eq:4omega}) reduces to two second order equations for $\omega _1 ^2$ 
and $\omega _2 ^2$ which can be solved easily. The characteristic equation reads 
\begin{equation}\label{eq:chartsym}
  0 = p_0^2 \, \big( p_0^2 - \omega _1 (\vec{p})^2 \big) \, 
  \big( p_0^2 - \omega _2 (\vec{p})^2 \big)
\end{equation}
with
\begin{equation}\label{eq:rootstsym}
  \omega _{1/2} (\vec{p})^2 = \frac{1}{4} p_{\rho} p_{\sigma} L^{\rho \sigma} \pm
  \sqrt{ \frac{1}{4} p_{\rho} p_{\sigma} p_{\tau} p_{\lambda} \big(
  L^{\rho \sigma \tau \lambda} - \frac{1}{4} L^{\rho \sigma} L^{\tau \lambda} \big) }
\end{equation}
Thus, the necessary and sufficient condition for hyperbolicity in the time-symmetric
case is that the right-hand side of (\ref{eq:rootstsym}) is real and non-negative,
i.e. $p_{\rho} p_{\sigma} L^{\rho \sigma} \, \ge \, 0 $ and
\begin{equation}\label{eq:hyptsym}
 \frac{1}{2}  \big( p_{\rho} p_{\sigma} L^{\rho \sigma} \big) ^2 
 \: \ge \:
 p_{\rho} p_{\sigma} p_{\tau} p_{\lambda} 
 L^{\rho \sigma \tau \lambda} 
 \: \ge \:  
 \frac{1}{4} \big( p_{\rho} p_{\sigma} L^{\rho \sigma} \big) ^2 
\end{equation}
for all $\vec{p}$ in $\mathbb{R}^3$. If, in addition, we want to prohibit 
zero-frequency modes, we have to strengthen the condition $p_{\rho} p_{\sigma} 
L^{\rho \sigma} \, \ge \, 0 $ to 
\begin{equation}\label{eq:nozf}
p_{\rho} p_{\sigma} 
L^{\rho \sigma} \, > \, 0 \qquad \text{for} \: \text{all} \quad \vec{p} \neq \vec{0} \;. 
\end{equation}
It would be desirable to rewrite (\ref{eq:hyptsym}) and (\ref{eq:nozf}) as
conditions on the constitutive matrix (\ref{eq:defM}). However, it is hard to 
see how this can be done in a practicable way.

Equation (\ref{eq:chartsym}), with (\ref{eq:rootstsym}) and (\ref{eq:nozf}), 
is the general form of the characteristic equation for the case that we have 
two real double-cones which are mirror-symmetric with respect to time inversion
and that zero-frequency modes are prohibited. 
Each of the two double-cones is the null cone of a \emph{Finsler metric}
\begin{equation}\label{eq:finsler}
  g_A ^{ab} \; = \; \frac{1}{2}
  \frac{\partial ^2 \big( 
  - p_0^2 + \omega _A (\vec{p})^2 \big)}{\partial p_a \partial p_b} \: , 
  \qquad A=1,2 \; . 
\end{equation}
More explicitly, the time-time, time-space and space-space components of
the two Finsler metrics $g_1$ and $g_2$ read
\begin{equation}\label{eq:finsler2}
  g_{1/2}^{00} \, = \, - \, 1 \; , \qquad
  g_{1/2}^{0 \mu} \, = \, 0 \; , \qquad
  g_{1/2}^{\rho \sigma} \, = \, \frac{1}{4} L^{\rho \sigma} \, \pm \,
  \frac{\sqrt{ \frac{1}{4} p_{\rho} p_{\sigma} p_{\tau} p_{\lambda} \big(
  L^{\rho \sigma \tau \lambda} - \frac{1}{4} L^{\rho \sigma} L^{\tau \lambda} \big) }}{
  L^{\mu \nu} p_{\mu} p_{\nu}} \; L^{\rho \sigma} 
\; .
\end{equation}
The conditions (\ref{eq:hyptsym}) and (\ref{eq:nozf}) guarantee that each of
these two Finsler metrics has Lorentzian signature at all $\vec{p} \neq \vec{0}$.
For such Finsler metrics a Fermat principle was proven in \cite{Perlick2006}.
Note, however, that in this article the Finsler light cones where assumed to
be smooth everywhere (except, of course, at the vertex). This is not the case
with the metrics (\ref{eq:finsler2}). At points where the square-root in 
(\ref{eq:finsler2}) has an isolated zero the two light cones form conical
singularities. The resulting phenomenon of ``conical refraction'' was
already mentioned in Section \ref{sec:cone}.

\section{Condition of non-birefringence}\label{sec:non-birefringence}
\noindent
If there are two different real roots $\omega _1 (\vec{p} ) \ge 0$ and 
$\omega _2 (\vec{p} ) \ge 0$, we have birefringence in the forward direction; similarly, 
if there are two different real roots $\omega _3 (\vec{p} ) \le 0$ and $\omega _4 (\vec{p} ) 
\le 0$, we have birefringence in the past direction. In this section  we want
to investigate the condition for non-birefringence. To that end we consider
the case that the four roots $\omega _1$, $\omega _2$, $\omega _3$, and $\omega _4$
pairwise coincide, $\omega _1 (\vec{p}) = \omega _2 (\vec{p} )$ and 
$\omega _3 (\vec{p}) = \omega _4 (\vec{p})$ for all $\vec{p} \in \mathbb{R} ^3$. 
For the time being we do not require that the roots are real. With the help
of our assumption that the roots pairwise coincide, it is easy to solve the
first two equations of (\ref{eq:4omega}),
\begin{equation}\label{eq:omega13}
  \omega _{1/3} (\vec{p}) \; = \;
  \frac{1}{4} \, p_{\rho} L^{\rho} \, \pm \,
  \sqrt{\, \frac{1}{4} \, p_{\rho} \, p_{\sigma} \, \big( L^{\rho \sigma} \, - \, 
  \frac{1}{4} L^{\rho} L^{\sigma} \big) \, } \: .
\end{equation}
By a generalised Galilean transformation we can always transform $L^{\rho}$ to zero, recall
(\ref{eq:4GalL}), so that (\ref{eq:omega13}) simplifies to 
\begin{equation}\label{eq:omega13red}
  \omega _{1/3} (\vec{p}) \; = \; \pm \,
  \sqrt{\, \frac{1}{4} \, p_{\rho} \, p_{\sigma} \,  L^{\rho \sigma} \, } \: .
\end{equation}
Note that time-symmetry is then automatically satisfied.
As a consequence, the reduced characteristic equation reads
\begin{equation}\label{eq:nobiref}
  (p_0-\omega _1 (\vec{p}) )^2 \, (p_0-\omega _3 (\vec{p}) )^2 
  \; = \;
  ( p_0 ^2 - \frac{1}{4}L^{\rho \sigma} p_{\rho} p_{\sigma} )^2 
  \; = \; 0 \; .  
\end{equation}
Thus, the characteristic variety is the null-cone of a quadratic form, i.e.,
our assumption that the roots pairwise coincide excludes proper Finsler
structures. The coefficients $L^{\rho \sigma}$ that determine the quadratic
form depend on $x$ and, in the case of a non-linear constitutive law, also 
on $\vec{D}(x)$ and $\vec{B}(x)$.  This result is true independent of whether 
or not we require hyperbolicity, i.e., independent of whether or not the 
roots (\ref{eq:omega13red}) are real.

By (\ref{eq:omega13red}), hyperbolicity is satisfied if and only $L^{\rho \sigma}$ is positive
semidefinite. This gives a Lorentzian or a degenerate quadratic form. The 
degenerate case is excluded if we require that $\omega _A 
(\vec{p}) \neq 0$ for $A=1,2,3,4$ and $\vec{p} \neq \vec{\, 0}$. Thus, the
condition of non-birefringence necessarily leads to a Lorentzian null cone 
if we require hyperbolicity and exclude zero-frequency modes. 

These findings corroborate earlier results found by Hehl and L{\"a}mmerzahl 
\cite{HehlLaemmerzahl2004} and by Itin \cite{Itin2005} for
linear constitutive laws. They give a satisfactory answer to the question
of what kind of light cones are possible in the case of non-birefringence. 
However, it would also be desirable to have a condition on the constitutive 
matrix  that is necessary and sufficient for non-birefringence. Such a condition
is still to be found.

\section{The symmetric hyperbolic case}\label{sec:symhyp}
\noindent
Recall that the evolution equations are symmetric hyperbolic if and 
only if there is a matrix $\mathbf{S}$ such that 
\begin{equation}\label{eq:symhyp}
  \mathbf{S}^{-1} p_{\rho} \mathbf{L}^{\rho} \mathbf{S} =
  \big( \mathbf{S}^{-1} p_{\rho} \mathbf{L}^{\rho} \mathbf{S} \big)^T
\end{equation}
for all $\vec{p}$ in $\mathbb{R}^3$, where $( \cdot )^T$ denotes 
transposition. With $\mathbf{L}^{\rho}$ from (\ref{eq:Lrho}), (\ref{eq:symhyp})
takes the form
\begin{equation}\label{eq:Msym} 
  \begin{pmatrix}
    0 & - p_{\rho} \mathbf{A}^{\rho} \\
    p_{\rho} \mathbf{A}^{\rho}  & \mathbf{0}\\
  \end{pmatrix}
  \; \mathbf{M} \; \mathbf{S} \; \mathbf{S}^T 
  \; = \; 
   \mathbf{S} \; \mathbf{S}^T \; \mathbf{M}^T \;  
  \begin{pmatrix}
    0 & - p_{\rho} \mathbf{A}^{\rho} \\
    p_{\rho} \mathbf{A}^{\rho}  & \mathbf{0}\\
  \end{pmatrix}
  \; .
\end{equation}  
We want to give a characterization of the symmetric hyperbolic case in
terms of the constitutive matrix $\mathbf{M}$. To that end we use the
following result which is based on a simple Schur lemma type argument.
\begin{lem}\label{lem:Schur}
  If $\mathbf{U}$ and $\mathbf{V}$ are real $3 \times 3$ matrices such 
  that
\begin{equation}\label{eq:Schur}
  p_{\rho} \mathbf{A}^{\rho} \mathbf{U}=   
  \mathbf{V} p_{\rho} \mathbf{A}^{\rho} 
\end{equation}
for all $\vec{p} \in \mathbb{R}^3$, then $\mathbf{U}=\mathbf{V}= 
c \mathbf{1}$ with some $c \in \mathbb{R}$.
\end{lem}
\begin{proof}
  For any $\vec{p} \neq \vec{\, 0}$ in $\mathbb{R}^3$, the matrix $p_{\rho}
  \mathbf{A}^{\rho}$ has a one-dimensional kernel spanned by $\vec{p}$.
  Thus, by applying (\ref{eq:Schur}) to multiples of $\vec{p}$ we see
  that $\mathbf{U}$ maps every one-dimensional subspace into itself, hence
  $\mathbf{U} = c\mathbf{1}$. Then (\ref{eq:Schur}) takes the form 
  $( \mathbf{V} - c \mathbf{1}) p_{\rho} \mathbf{A}^{\rho} = \mathbf{0}$.
  If $\vec{p}$ runs over $\mathbb{R}^3$, the image of $p_{\rho} \mathbf{A}^{\rho}$
  runs over $\mathbb{R}^3$. Thus, our last equation requires $\mathbf{V} = 
  c \mathbf{1}$    
\end{proof}
With the help of this lemma, we prove the following proposition.
\begin{proposition}\label{prop:posdef}
The evolution equations are symmetric hyperbolic if and only if the
constitutive matrix $\mathbf{M}$ is (i) the zero matrix, (ii) positive
definite, or (iii) negative definite.
\end{proposition}
\begin{proof}
By applying Lemma \ref{lem:Schur} to each of the four $3 \times 3$ blocks
of (\ref{eq:Msym}) we find 
\begin{equation}\label{eq:symblock1}
  \mathbf{M} \mathbf{S} \mathbf{S}^T =
  \begin{pmatrix}
     b \mathbf{1} & d \mathbf{1} \\
     a \mathbf{1} & c \mathbf{1} 
  \end{pmatrix}
  \; ,
\end{equation}
\begin{equation}\label{eq:symblock2}
  \mathbf{S} \mathbf{S}^T \mathbf{M}^T =
  \begin{pmatrix}
     c \mathbf{1} & -a \mathbf{1} \\
     -d \mathbf{1} & b \mathbf{1} 
  \end{pmatrix}
  \; ,
\end{equation}
with real numbers $a,b,c,d$. As the left-hand side of (\ref{eq:symblock1})
is the transpose of the left-hand side of (\ref{eq:symblock2}), this can
be true only if $a=d=0$ and $b=c$, i.e. the necessary and sufficient condition
for symmetric hyperbolicity is that there exists an invertible matrix $\mathbf{S}$
such that $\mathbf{M} = c (\mathbf{S} \mathbf{S}^T)^{-1}$. This is true (i) with
$c=0$ if and only if $\mathbf{M}=\mathbf{0}$, (ii) with $c > 0$ if and only if
$\mathbf{M}$ is positive definite, and (iii) with $c<0$ if and only if $\mathbf{M}$
is negative definite.
\end{proof}
Clearly, the case $\mathbf{M}= \mathbf{0}$ yields the characteristic equation $p_0^6 = 0$
and is physically uninteresting. Thus, $\mathbf{M}$ must be positive or negative
definite to give symmetric hyperbolic evolution equations. In the case of a linear 
constitutive law this condition is equivalent to the assumption that the 
\emph{energy density} $w = \frac{1}{2} ( E_{\rho}D^{\rho} + {\mathcal{H}}_{\rho} B^{\rho} )$ 
is positive or negative definite. In particular, the positive (or negative) 
definiteness of $\mathbf{M}$ requires that $\boldsymbol{\kappa}$ and 
$\boldsymbol{\nu}$ are positive (or negative) definite, i.e., that we have  
positive (or negative) definite permittivity $\boldsymbol{\varepsilon} = 
\boldsymbol{\kappa}^{-1}$ and permeability $\boldsymbol{\mu} = \boldsymbol{\nu}^{-1}$.   

Proposition \ref{prop:posdef} generalises a result that was derived in 
\cite{Perlick2000b}, Section 2.1. There only linear constitutive laws were 
considered and a Lorentzian metric was presupposed. It was shown that, if 
magneto-electric cross-terms are absent and permeability and permittivity 
are positive definite, the evolution equations are symmetric hyperbolic. 
The above result shows that the definiteness condition is not 
only sufficient but also necessary and that this result (i) carries over 
to non-linear local constitutive laws, (ii) can be formulated without 
reference to a background metric, and (iii) remains true if magneto-electric
cross-terms are allowed.  

\section{Examples}\label{sec:ex}
\subsection{Biisotropic media}\label{sec:biiso}
\noindent
A medium is called \emph{biisotropic} (at a point $x$) if there is coordinate system
such that each $3 \times 3$ block of the constitutive matrix is a scalar multiple
of the unit matrix (at $x$), i.e.  
\begin{equation}\label{eq:biiso}
\mathbf{M} \, = \, 
\begin{pmatrix}
\boldsymbol{\kappa} & \boldsymbol{\chi}
\\
\boldsymbol{\gamma} & \boldsymbol{\nu}
\end{pmatrix}
\, = \, 
\begin{pmatrix}
\kappa \mathbf{1} & \chi \mathbf{1}
\\
\gamma \mathbf{1} & \nu \mathbf{1} 
\end{pmatrix}
\end{equation}
with scalars $\kappa, \nu, \chi, \gamma$. If this is true with $\chi = \gamma = 0$,
the medium is called \emph{isotropic}. We want to find a neccessary and
sufficient condition for a biisotropic medium to yield hyperbolic evolution 
equations.

With (\ref{eq:biiso}), the characteristic equation becomes
\begin{equation}\label{eq:charbiiso}
\mathrm{det} \big( \,
p_0^2 \, \mathbf{1} \, - \, p_0 \, p_{\rho} \, \mathbf{A} ^{\rho} (\chi - \gamma ) 
\, + \, 
p_{\rho} \, p_{\sigma} \, \mathbf{A} ^{\rho} \mathbf{A} ^{\sigma} (\kappa \nu - \chi \gamma )
\, \big) \, = \, 0
\end{equation}
as can be read from (\ref{eq:char2}) or (\ref{eq:char1}). The determinant can 
easily be calculated, resulting in
\begin{equation}\label{eq:charbiiso2}
p_0^2 \, \Big( \, \big( \, p_0^2 \, - \, | \, \vec{p} \, | ^2 \, 
(\kappa \nu - \chi \gamma ) \, \big)^2
\, + \, p_0^2 \, | \, \vec{p} \, | ^2 \, (\chi - \gamma )^2 \, \Big) 
\, = \, 0 \; .
\end{equation}
Hyperbolicity requires that all roots of this equation are real, which is obviously
the case if and only if $\chi = \gamma$ and $\kappa \nu - \chi \gamma \ge 0 \,$.
In the last inequality we replace the $\ge$ sign by a $>$ sign, as the $=$ sign 
only gives zero-frequency modes. With this unphysical zero-frequency case omitted,
we can thus say that a biisotropic medium yields hyperbolic evolution equations
if and only if
\begin{equation}\label{eq:biisocon}
\chi = \gamma \quad \text{and} \quad \mathrm{det} ( \mathbf{M} )  > 0 \; .
\end{equation}
In this case the reduced characteristic equation has two roots of multiplicity 2,
\begin{equation}\label{eq:biisoroots}
\omega _1 ( \vec{p} \, ) \, = \, \omega _2 ( \vec{p} \, ) \, = \, 
- \, \omega _3 ( \vec{p} \, ) \, = \, - \, \omega _4 ( \vec{p} \, ) \, = \, 
| \, \vec{p} \, | \, \sqrt{\mathrm{det} ( \mathbf{M} ) \,} \; .
\end{equation}
There is no birefringence, i.e., we have a unique future light cone and a unique
past light cone. This double-cone is the null cone of the Lorentzian metric
\begin{equation}\label{eq:biisolor}
g^{00} \, = \, - \, 1  \, \quad g^{0 \mu} \, = \, 0 \, \quad 
g^{\mu \nu} \, = \, \sqrt{\mathrm{det} ( \mathbf{M} )} \; \delta ^{\mu \nu} 
\end{equation}
and, of course, also of any metric that is conformal to this one.

Note that (\ref{eq:biisocon}) is equivalent to the requirement that the matrix
(\ref{eq:biiso}) is (positive or negative) definite. From Section \ref{sec:symhyp}
we know that then the evolution equations are symmetric hyperbolic. Hence, for 
a biisotropic medium the condition of hyperbolicity is equivalent to the 
condition of symmetric hyperbolicity. For an isotropic medium (\ref{eq:biisocon})
reduces to the condition that $\kappa = \varepsilon ^{-1}$
and $\nu = \mu ^{-1}$ must have the same sign.



\subsection{Born-Infeld theory}\label{sec:borninfeld}
\noindent
Born-Infeld theory was introduced by Born and Infeld in 1934 \cite{BornInfeld1934}.
The motivation was to modify standard vacuum electrodynamics in such a way that the
field energy in a small ball around a point charge is finite. This was achieved by
assuming a non-linear vacuum constitutive law of the form $H_{ab} = \partial L / 
\partial F_{ab}$, where $L$ is the Born-Infeld Lagrangian $L = -\sqrt{b^4 + b^2 F_{ab}F^{ab}
- *F_{ab}F^{ab}}$. Here one assumes that, as in standard vacuum electrodynamics, a
spacetime metric of Lorentzian signature is given: the star is the Hodge operator defined 
by the spacetime metric and latin indices are raised and lowered with the spacetime metric.
$b$ is a constant of nature, called the ``absolute field'' by Born and Infeld.    

As we allowed for non-linear constitutive laws throughout, Born-Infeld
theory fits perfectly well into the general scheme considered in this
paper. In the following we will apply the results of the preceding sections
to Born-Infeld theory, thereby deriving the structure of the characteristic
variety (i.e., of the light cones) in the Born-Infeld theory and establishing
the result that the Born-Infeld theory admits a well-posed initial-value
problem. None of these results is new. (The light cones  of the Born-Infeld theory
were determined, e.g by Boillat \cite{Boillat1970}; a proof that the Born-Infeld
initial-value problem is well-posed can be found e.g. in Serre \cite{Serre2004}.) 
However, the derivations given here are quite different from the ones available
in the literature and illustrate the general results given above.

Our first goal is to demonstrate that, for initial values given on a 
hypersurface that is spacelike with respect to the spacetime metric,
the Born-Infeld initial-value problem is well-posed. To that end we
choose coordinates $(x^0,x^1,x^2,x^3)$ such that the chosen spacelike 
hypersurface is given by the equation $x^0 = \mathrm{constant}$. In addition,
we may assume that, at some particular point on the hypersurface, the coordinates 
are pseudo-orthonormal with respect to the  spacetime metric. This leaves he 
freedom of orthogonal transformations of the spatial coordinates on the tangent
space of the chosen point. Then, at the chosen point, the Lagrangian takes the form 
$L= -\sqrt{b^4 +b^2(B^{\mu}B_{\mu} - E^{\mu}E_{\mu}) - (E^{\mu}B_{\mu})^2}$. Here
and in the following, greek indices are raised and lowered with the Kronecker delta.
At the chosen point, the constitutive law reads
\begin{gather}
\label{eq:conBI1}
D^{\mu} = \frac{\partial L}{\partial E_{\mu}} =
\frac{b^2 E^{\mu} + B^{\nu}E_{\nu} B^{\mu}}{
\sqrt{b^4 +b^2(B^{\mu}B_{\mu} - E^{\mu}E_{\mu}) - (E^{\mu}B_{\mu})^2}} \: ,      
\\
\label{eq:conBI2}
\mathcal{H}{}^{\mu} = - \frac{\partial L}{\partial B_{\mu}} =
\frac{b^2 B^{\mu} - B^{\nu}E_{\nu} E^{\mu}}{
\sqrt{b^4 +b^2(B^{\mu}B_{\mu} - E^{\mu}E_{\mu}) - (E^{\mu}B_{\mu})^2}}     \: .
\end{gather}
Clearly, in the limit $b \to \infty$ the non-linear equations (\ref{eq:conBI1}) and 
(\ref{eq:conBI2}) tend to the linear standard vacuum constitutive law $D^{\mu} = E^{\mu}$ 
and $\mathcal{H}{}^{\mu} = E^{\mu}$. (When comparing our equations (\ref{eq:conBI1}) and 
(\ref{eq:conBI2}) with the corresponding equations on page 437 in the original Born-Infeld 
paper \cite{BornInfeld1934}, note that there is a sign error in the latter.) 

Equations (\ref{eq:conBI1}) and (\ref{eq:conBI2}) can be solved for $E^{\mu}$ and $\mathcal{H}{}^{\mu}$, which 
demostrates that our coordinates are admissible in the sense of Definition \ref{def:regular}.
The resulting equations, which are found after an elementary though rather tedious calculation,
read
\begin{gather}\label{eq:conBI3}
E^{\mu} = \frac{\partial W}{\partial D_{\mu}} =
\frac{1}{W} \big( (b^2 + B^{\nu}B_{\nu})D^{\mu} - B^{\nu}D_{\nu} B^{\mu}\big)
\\
\label{eq:conBI4}
\mathcal{H}{}^{\mu} = \frac{\partial W}{\partial B_{\mu}} =
\frac{1}{W} \big( (b^2 + D^{\nu}D_{\nu})B^{\mu} - B^{\nu}D_{\nu} D^{\mu} \big)
\end{gather}
 where
\begin{equation}\label{eq:defW}
W(\vec{D}, \vec{B}) = 
\sqrt{(b^2 +B^{\rho}B_{\rho})(b^2+D^{\sigma}D_{\sigma}) - (B^{\tau}D_{\tau})^2}
\end{equation}
is the Legendre transform of $L(\vec{E}, \vec{B})$ with respect to the pair of 
variables $\vec{E}, \vec{D}$, i.e., $W = E_{\mu}D^{\mu}-L$.  Now the $3 \times 3$
blocks of the constitutive matrix take the form
\begin{gather}
\nonumber
\kappa _{\alpha \beta} = \frac{\partial ^2 W}{\partial D^{\alpha} \partial D^{\beta}} 
\qquad \qquad \qquad \qquad
\\
\nonumber
= \frac{(b^2+B^{\sigma}B_{\sigma})}{W^3}  \Big(
W^2 \delta _{\alpha \beta} - (b^2+B^{\rho}B_{\rho}) D_{\alpha}D_{\beta}   
- (b^2+D^{\tau}D_{\tau}) B_{\alpha}B_{\beta} + B^{\lambda}D_{\lambda}
(B_{\alpha}D_{\beta}+ B_{\beta}D_{\alpha}) \Big) \; ,
\\
\label{eq:MBI}
\nu_{\alpha \beta} = \frac{\partial ^2 W}{\partial B^{\alpha} \partial B^{\beta}}
= \frac{(b^2+D^{\tau}D_{\tau})}{(b^2+B^{\rho}B_{\rho})}
\kappa_{\alpha \beta} \; , \qquad \qquad 
\\
\nonumber
\gamma _{\alpha \beta} = \chi _{\beta \alpha} =
\frac{\partial ^2 W}{\partial B^{\alpha} \partial D^{\beta}}
= \frac{B^{\tau}D_{\tau}}{(b^2+B^{\rho}B_{\rho})}
\kappa _{\alpha \beta} + \frac{B_{\alpha}D_{\beta}-B_{\beta}D_{\alpha}}{W} \; .
\qquad 
\end{gather}
$\boldsymbol{\kappa}$ is symmetric, so it has three real eigenvalues with 
orthogonal eigenvectors. As our coordinate system is fixed only up to 
orthogonal transformations of the spatial coordinates, we may choose the
coordinates such that $\boldsymbol{\kappa}$ is diagonal.
The eigenvalues $\kappa _1, \kappa _2$ and $\kappa _3$ of $\boldsymbol{\kappa}$ 
are

\begin{gather}
\nonumber
\kappa _{1/2} = 
\frac{b^2(b^2+B^{\sigma}B_{\sigma})}{W^3}
\Big( b^2 + \frac{D^{\mu}D_{\mu}+B^{\nu}B_{\nu}}{2} \pm
\sqrt{\frac{(B^{\rho}B_{\rho}-D^{\sigma}D_{\sigma})^2}{4}
+ (D^{\tau}B_{\tau})^2 \;} \: \Big) \; ,
\\
\label{eq:kappaev}
\kappa _3 = \frac{(b^2+B^{\sigma}B_{\sigma})}{W} \; .
\qquad \qquad
\end{gather}
As they are strictly positive, $\boldsymbol{\kappa}$ is positive 
definite for all $(\vec{D}, \vec{B})$ in $\mathbb{R} ^6$.

Now we make a coordinate transformation $(x^0,x^1,x^2,x^3) \mapsto 
(\tilde{x}{}^0, \tilde{x}{}^1, \tilde{x}{}^2, \tilde{x}{}^3 )$ 
that induces at the chosen point a generalised Galilean 
transformation (\ref{eq:Galileo1}) and (\ref{eq:Galileo2}) with
\begin{equation}\label{eq:BIgal}
\mathbf{a} = \mathbf{b}{}^{-1} = \mathrm{diag}( \sqrt{\kappa _1},
\sqrt{\kappa _2},\sqrt{\kappa _3}) \; , \quad
c = \mathrm{det}(\mathbf{a}) = \sqrt{\kappa _1 \kappa _2 \kappa _3} \; , \quad
v^{\sigma} = \frac{1}{W} \epsilon ^{\sigma \mu \nu}B_{\mu}D_{\nu} \; .
\end{equation}
By (\ref{eq:GalM}), this transforms (\ref{eq:MBI}) into
\begin{equation}\label{eq:MBI2}
\tilde{\kappa}{} _{\alpha \beta} = \delta _{\alpha \beta} \, , \quad
\tilde{\nu}{}_{\alpha \beta} = \frac{(b^2+D^{\tau}D_{\tau})}{(b^2+B^{\rho}B_{\rho})}
\delta _{\alpha \beta} \; , \quad 
\tilde{\gamma}{}_{\alpha \beta} = \tilde{\chi}{}_{\beta \alpha} =
\frac{B^{\tau}D_{\tau}}{(b^2+B^{\rho}B_{\rho})}
\delta _{\alpha \beta}  \; .
\end{equation}
By (\ref{eq:transM}), the constitutive matrix takes the form
\begin{equation}\label{eq:MBI3}
  \tilde{\mathbf{M}} \,
  \; = \;
  \begin{pmatrix} 
    \: \mathbf{1} \: & \:  \frac{B^{\rho}D_{\rho}}{b^2+B^{\tau}B_{\tau}} \mathbf{1}    \:
    \\[0.1cm]
    \mathbf{0} & \mathbf{1}
  \end{pmatrix}
  ^T  
  \;
  \begin{pmatrix} 
    \: \mathbf{1} \: 
    &  \mathbf{0} 
    \\[0.1cm]
    \mathbf{0} & \: \frac{W^2}{(b^2+B^{\lambda}B_{\lambda})^2} \mathbf{1} \:
  \end{pmatrix}
  \begin{pmatrix} 
    \: \mathbf{1} \: & \:  \frac{B^{\rho}D_{\rho}}{b^2+B^{\tau}B_{\tau}} \mathbf{1}    \:
    \\[0.1cm]
    \mathbf{0} & \mathbf{1}
  \end{pmatrix}
  \; .
\end{equation}
As this matrix is obviously positive definite, Proposition 
\ref{prop:posdef} proves that, in the twiddled coordinates, the evolution equations 
are symmetric hyperbolic, so the initial-value problem is, indeed, well-posed.

Note that the coordinate transformation was necessary for achieving our goal. The 
Galilean boost with $v^{\sigma}$ had the effect of killing the antisymmetric part 
of $\boldsymbol{\gamma} = \boldsymbol{\chi} ^T$. In contrast to $\tilde{\mathbf{M}}$, 
the original constitutive matrix $\mathbf{M}$ was not positive definite for all values
of $(\vec{D},\vec{B})$ but only for $(\vec{D},\vec{B})$ in a certain neighborhood of
the origin in $\mathbb{R} ^6$.

Finally, we want to calculate the Born-Infeld light cones. 
In the twiddled coordinates, the characteristic equation (\ref{eq:char2}) takes the form
\begin{equation}\label{eq:charBI}
0 \, = \, \mathrm{det} \Big( \, \tilde{p}{}_0^2 \mathbf{1} + 
\frac{W^2}{(b^2+B^{\mu}B_{\nu})^2}
\tilde{p}{}_{\rho} \tilde{p}{}_{\sigma} \mathbf{A} ^{\rho} \mathbf{A} ^{\sigma} \, \Big)
\, = \,
\tilde{p}{}_0^2 \, 
\big( \, \tilde{p}{}_0^2- 
\frac{W^2 ( \tilde{p}{}_1^2 + \tilde{p}{}_2^2 + \tilde{p}{}_3^2 )}{(b^2+B^{\mu}B_{\mu})^2} 
 \, \big)^2
\; .
\end{equation}
Thus, there is no birefringence (cf. Boillat \cite{Boillat1970}); the Born-Infeld theory 
determines a unique past and a unique future light cone, given by the equation
\begin{equation}\label{eq:BIcone1}
\tilde{G}{}^{ab} \tilde{p}{}_a \tilde{p}{}_b \, = \,
- \tilde{p}{}_0^2 + 
\frac{W^2 ( \tilde{p}{}_1^2 + \tilde{p}{}_2^2 + \tilde{p}{}_3^2)}{(b^2+B^{\mu}B_{\mu})^2} 
 \, = \, 0
 \; .
\end{equation}
In the original coordinates, in which the light cone of the spacetime metric takes the 
form $p_0^2 = p_1^2 + p_2^2 + p_3^2$, (\ref{eq:BIcone1}) reads
\begin{equation}\label{eq:BIcone2}
G^{ab} p_a p_b \, = \,
- \frac{ (p_0 - v^{\rho} p_{\rho})^2}{\kappa _1 \kappa _2 \kappa _3} \, + \, 
 \frac{W^2}{(b^2+B^{\mu}B_{\mu})^2} 
\, \Big( \, 
\frac{p_1^2}{\kappa _1} + \frac{p_2^2}{\kappa _2} + \frac{p_3^2}{\kappa _3}
\, \Big) 
\, = \, 0
\; .
\end{equation}
With the $v^{\rho}$ from (\ref{eq:BIgal}) and the $\kappa _{\mu}$ from
(\ref{eq:kappaev}) the equation for the Born-Infeld light cones becomes
\begin{equation}
\frac{(W \, p_0 - \epsilon^{\rho \mu \nu} p_{\rho}B_{\mu}D_{\nu})^2}{
b^2(b^2+B^{\lambda}B_{\lambda})} \, = \, 
b^2 (p_1^2+p_2^2+p_3^2) \, 
- \frac{D^{\tau}D_{\tau}+B^{\kappa}B_{\kappa}}{2} 
(p_1^2+p_2^2) \, - \, 
\sqrt{\frac{(B^{\rho}B_{\rho}-D^{\sigma}D_{\sigma})^2}{4}
+ (D^{\tau}B_{\tau})^2 \;} \: (p_1^2-p_2^2) \; \Big) \; .
\end{equation}
 
\section*{Conclusions}
In this article we have considered Maxwell's equations with a local 
constitutive law and we have found some useful results. In particular,
we have derived several versions of the characteristic equation and 
we have worked out a method of how to calculate its roots; 
moreover, we have conveniently characterised the class of all 
constitutive laws that give symmetric hyperbolic evolution equations.
However, symmetric hyperbolicity is not necessary
for well-posedness of the initial-value problem. If we want to
characterise the class of all constitutive laws for which the initial-value
problem is well-posed, we need a criterion for strong hyperbolicity.
This is an open problem. It would also be desirable to characterise
all constitutive laws that give hyperbolic evolution equations. Again, 
this is an open problem. 
We were able to characterise the
light cones in the case of invariance under temporal or spatial 
inversions and in the case of birefringence; however, we could
not find a condition on the constitutive matrix that is necessary 
and sufficient for either of these two properties. We have found a 
certain group of transformations that act on the set of all constitutive 
matrices and leave the characteristic equation invariant; however, 
we could not determine the set of all such transformations. So there
are a lot of open problems that should be addressed in future work.
\\[0.2cm]
Note added in proof: After this paper was submitted the author learned about
Schuller, Witte and Wohlfarth \cite{SchullerWitteWohlfarth2010},
R{\"a}tzel, Rivera and Schuller \cite{RatzelRiveraSchuller2011}, and
Favaro and Bergamin \cite{FavaroBergamin2011}
where important related results were found.


\section*{Acknowledgment}
I have profited very much from discussions with Friedrich Hehl,
Yuri Obukhov and Yakov Itin on the pre-metric approach to
electrodynamics, and from seminars on the subject with
Robin Tucker, David Burton, Jonathan Gratus and other
colleagues in Lancaster.



\end{document}